\documentclass[aip,jmp,amsmath,amssymb,reprint]{revtex4-2}
\usepackage{mathptmx}

\usepackage{amsthm,mathtools}
\usepackage{enumerate}
\usepackage{url}
\newtheorem{theorem}{Theorem}[section]
\newtheorem{lemma}[theorem]{Lemma}
\newtheorem{proposition}[theorem]{Proposition}
\newtheorem{corollary}[theorem]{Corollary}
\theoremstyle{definition}

\theoremstyle{remark}
\newtheorem{remark}[theorem]{Remark}

\newcommand{\cH}{\mathcal{H}}

\newcommand{\rank}{\operatorname{rank}}
\newcommand{\diag}{\operatorname{diag}}
\newcommand{\dd}{\mathrm{d}}
\newcommand{\Id}{\mathrm{I}}

\begin{document}

\title{Sharp Entropy Cost of Coherent Leakage Across a Block Decomposition}

\author{Hassan Nasreddine}
\email{hassan.nasreddine@fortisec.net}
\affiliation{FortiSec Systems Inc., Canada}
\date{15 October 2026}

\begin{abstract}
Let $\mathcal{H}=P\mathcal{H}\oplus Q\mathcal{H}$ and let $\Pi\rho=P\rho P\oplus Q\rho Q$ be the two-block pinching of a density matrix. We study the block-coherence entropy $D(\rho\|\Pi\rho)$, the relative-entropy cost of removing cross-block coherence. First, we prove a midpoint Bogoliubov--Kubo--Mori (BKM) coercivity estimate, giving the operator-kernel lower bound $D(\rho\|\Pi\rho)\ge \operatorname{Tr}[B^*\Omega_{A,C}^{-1}(B)]$ for strictly positive diagonal blocks. The proof uses the block-sign involution $U=P-Q$ to obtain a sign-symmetric midpoint lower bound for the BKM Hessian. Second, we solve the associated constrained entropy-minimization problem: at fixed protected-sector floor, leakage mass, and quadratic coherence, the minimum can be achieved by a single active two-level block, independently of the ambient dimension. Third, under a uniform floor $A\ge a_0P$ and $0<\varepsilon_Q\le a_0/2$, we recover the logarithmic boundary law $D(\rho\|\Pi\rho)\ge\|B\|_F^2\log(a_0/\varepsilon_Q)$, showing a divergent entropy cost per unit quadratic coherence as leakage mass tends to zero. As a finite-dimensional application, the estimates give an entropy certificate for coherent leakage from a protected subspace: $P\rho Q$ measures protected--leakage coherence, and the constrained entropy minimum is attained by a state supported on a single effective two-level transition. The BKM estimate also yields a dissipation lower bound along an idealized uniform block-dephasing orbit. Finally, the midpoint mechanism extends to finite-dimensional trace-preserving $\mathbb{Z}_2$-graded automorphisms and finite von Neumann algebras under bounded uniform positivity assumptions.
\end{abstract}

\maketitle

\section{Introduction}
\label{sec:intro}

\subsection{Physical setting and motivation}

Many quantum systems admit a physically distinguished \emph{block decomposition}
of their Hilbert space: a direct-sum structure $\mathcal{H}=P\mathcal{H}\oplus Q\mathcal{H}$
arising from a superselection rule, a charge sector boundary, a measurement partition,
a pointer-state decomposition, or the separation between a system and a leak (environment)
sector. In all these cases, the block-diagonal subalgebra
$\{P\rho P\oplus Q\rho Q:\rho\text{ density matrix}\}$
represents states with no quantum coherence across the boundary.

\emph{Block dephasing} is the map $\Pi\rho:=P\rho P\oplus Q\rho Q$ that
removes all cross-block coherence while preserving each block's diagonal data.
In quantum information, $\Pi$ is a measurement-induced or environment-induced
decoherence channel. The relative entropy
\[
D(\rho\|\Pi\rho):=\operatorname{Tr}[\rho(\log\rho-\log\Pi\rho)]
\]
is the \emph{block-coherence entropy}: it measures the entropy cost of removing
cross-block coherence from $\rho$, or equivalently the informational distance
from $\rho$ to its dephased shadow. More precisely, among all block-diagonal
states it is the minimum relative-entropy distance:
\[
D(\rho\|\Pi\rho)=
\inf_{\substack{\sigma\ge0,\ \operatorname{Tr}\sigma=1\\ \Pi\sigma=\sigma}}
D(\rho\|\sigma).
\]
Indeed, the exact Pythagorean identity~\eqref{eq:pythagorean} shows that
$D(\rho\|\sigma)=D(\rho\|\Pi\rho)+D(\Pi\rho\|\sigma)$;
thus the unique minimizing density matrix is $\sigma=\Pi\rho$.
In finite dimensions $\operatorname{supp}\rho\subseteq
\operatorname{supp}\Pi\rho$, so the minimum is finite; divergent
relative entropies for other candidates use the usual support convention.
Quantum relative entropy also admits an operational
distinguishability interpretation through asymptotic
quantum hypothesis testing~\cite{HiaiPetz1991}.

\paragraph{Terminology.}
Throughout this paper, \emph{block dephasing} or \emph{pinching} refers specifically
to the conditional expectation $\Pi(\rho)=P\rho P\oplus Q\rho Q$.
We do not use ``dephasing'' to denote a general dynamical noise model unless
explicitly stated. \emph{Coherent leakage} refers to the off-block component $B=P\rho Q$,
whose Frobenius norm squared $c=\|B\|_F^2$ we call the \emph{quadratic coherence}.
The \emph{leakage mass} is $\varepsilon_Q=\operatorname{Tr}(Q\rho Q)$.
This terminology is interpretive: the theorems are block-pinching entropy
inequalities for an arbitrary two-block decomposition. When $P\mathcal{H}$ is a
protected, computational, or metastable subspace, the protected/leakage language
provides the information-theoretic reading developed in Section~\ref{sec:application}.
The \emph{entropy cost} of coherent leakage is the quantity $D(\rho\|\Pi\rho)$.

\paragraph{Block structure.}
For a density matrix $\rho$, write
\[
\rho=\begin{pmatrix}A&B\\B^*&C\end{pmatrix},\qquad
A=P\rho P,\quad B=P\rho Q,\quad C=Q\rho Q.
\]
The map $\Pi$ is the conditional expectation onto the block-diagonal subalgebra
associated with the $P\oplus Q$ decomposition~\cite{Umegaki1962,Takesaki1979,AccardiCecchini1982}.
Set $d_P=\dim P\cH$, $d_Q=\dim Q\cH$.
Throughout Sections~\ref{sec:setup}--\ref{sec:dynamics} we work in finite dimensions and assume $A,C>0$
when invoking the Bogoliubov--Kubo--Mori (BKM) operator theorem; the variational theorem extends to
boundary states by continuity. Sections~\ref{sec:application} and~\ref{sec:z2} record the protected-subspace application and the graded and
finite-von-Neumann-algebra extensions of the midpoint mechanism.

\subsection{Finite-dimensional model case}

Before stating the main results, we record the core finite-dimensional case.
The full variational theorem (Theorem~\ref{thm:variational}) gives the exact
version of this statement.

\begin{quote}
\textbf{Guiding statement.}
\textit{Let $\mathcal{H}=P\mathcal{H}\oplus Q\mathcal{H}$ be finite-dimensional.
Among all density matrices $\rho$ with:
(i) protected-sector floor $P\rho P\ge a_0P$;
(ii) prescribed leakage mass $\operatorname{Tr}(Q\rho Q)=\varepsilon_Q$;
(iii) prescribed quadratic coherence $\|P\rho Q\|_F^2=c$;
the minimum block-coherence entropy $D(\rho\|\Pi\rho)$ equals
\[
\Phi(a_*,\varepsilon_Q,c),\qquad a_*=1-\varepsilon_Q-(d_P-1)a_0.
\]
Here $\Phi(a,\varepsilon,x)$ denotes the two-level relative entropy
between a positive $2\times2$ matrix and its diagonal part:
\[
\Phi(a,\varepsilon,x)=D\!\left(
\begin{pmatrix}a&\sqrt{x}\\\sqrt{x}&\varepsilon\end{pmatrix}
\middle\|\operatorname{diag}(a,\varepsilon)\right).
\]
Its domain is $a,\varepsilon\ge0$, $0\le x\le a\varepsilon$;
see the formal definition in Sec.~\ref{sec:setup} and
Lemma~\ref{lem:Phi-continuity}.
The minimum can be attained by a single active two-level block: a state of the form
$\bigl(\begin{smallmatrix}a_*&\sqrt{c}\\\sqrt{c}&\varepsilon_Q\end{smallmatrix}\bigr)$
on one pair of vectors $(|p_1\rangle,|q_1\rangle)$, with incoherent spectators
at the floor value $a_0$ in the remaining $P$-directions.}
\end{quote}

The one-channel principle states that the constrained minimum can always be
realized by a two-level configuration, regardless of the ambient dimension.
Thus the high-dimensional constrained problem collapses to a scalar $2\times2$
problem. The Coherence Compression Principle (Theorem~\ref{thm:compression})
provides the structural reason for this collapse.

Theorem~\ref{thm:variational} is the precise version of this statement.
The admissibility conditions $1-\varepsilon_Q\ge d_Pa_0$ and $c\le a_*\varepsilon_Q$
are exactly the feasibility conditions ensuring that the two-level optimizer
is a valid density matrix (Schur-complement positivity).

\subsection{Main results and solved extremal problem}
\label{ssec:main-results}

We present four named results. Detailed proofs follow in Sections
\ref{sec:BKM-geometry}--\ref{sec:dynamics}.

\paragraph*{Result A. Midpoint BKM coercivity}

Let $M=A\oplus C$ and $Y=\bigl(\begin{smallmatrix}0&B\\B^*&0\end{smallmatrix}\bigr)$.
The block-sign involution $U=P-Q$ satisfies $UMU^*=M$ and $UYU^*=-Y$.
Consequently, for any $t\in[0,1)$,
\[
H_{M+tY}(Y,Y)=H_{M-tY}(Y,Y),
\]
where $H_N(Z,Z)=\int_0^\infty\operatorname{Tr}[Z(N+rI)^{-1}Z(N+rI)^{-1}]\,dr$ is
the BKM quadratic form. The symmetry $H_{M+tY}(Y,Y)=H_{M-tY}(Y,Y)$, together
with the pointwise scalar resolvent inequality
\[
\tfrac{1}{2}\bigl((1+u)^{-2}+(1-u)^{-2}\bigr)\ge1,\quad|u|<1,
\]
applied fibrewise to the resolvent integral, then gives
\[
H_{M+tY}(Y,Y)\ge H_M(Y,Y).
\]
(Kubo--Ando concavity of operator means is used separately in
Appendix~\ref{app:petz} for the Petz-metric analogue; it is not part of the
BKM argument.)
Integrating this midpoint lower bound along the affine path $\rho(t)=M+tY$
from $t=0$ to $t=1$ yields the BKM operator-kernel lower bound (Theorem~\ref{thm:BKM-lower})
\[
D(\rho\|\Pi\rho)\ge\operatorname{Tr}[B^*\Omega_{A,C}^{-1}(B)].
\]
This is a \emph{coercivity} property specific to block off-diagonal
perturbations; it is not a general statement about affine monotonicity of the
BKM Hessian.

\paragraph*{Result B. One-channel extremizer}

The following is the main solved extremal problem.

\textbf{Theorem (informal).} \emph{At fixed protected-sector floor
$P\rho P\ge a_0P$, fixed leakage $\varepsilon_Q=\operatorname{Tr}(Q\rho Q)$, and fixed
quadratic coherence $c=\|P\rho Q\|_F^2$, the minimum block-coherence entropy equals
\[
\Phi(a_*,\varepsilon_Q,c),\qquad a_*=1-\varepsilon_Q-(d_P-1)a_0,
\]
and can be attained by a single active two-level block plus incoherent spectators.}

See Theorem~\ref{thm:variational} for the precise statement. For every feasible
$c>0$, the constraint $\|P\rho Q\|_F^2=c$ defines a non-convex feasible
set: a state and its block-sign conjugate have opposite coherence blocks
and obey the same constraints, whereas their average has zero off-block
coherence. Thus the one-channel minimum is not simply a consequence of
convex optimization; its key structural step is the coherence-compression
construction of Theorem~\ref{thm:compression}. A key consequence
is the \emph{one-channel principle}: for the constrained problem, the minimum
can be realized by a qubit-like two-level structure. Distributing the same scalar
coherence data over additional channels cannot lower the entropy cost.

The theorem proves existence of a one-channel optimizer. The complete equality
classification is left open.

\paragraph*{Result C. Logarithmic boundary barrier}

If $A\ge a_0P$ and $\varepsilon_Q\le a_0/2$, then
(Corollary~\ref{cor:log-bound})
\[
D(\rho\|\Pi\rho)\ge\|B\|_F^2\,\log\!\Bigl(\frac{a_0}{\varepsilon_Q}\Bigr).
\]
As $\varepsilon_Q\to0$, positivity ($c\le a_*\varepsilon_Q$) forces $c\to0$, so the
total entropy need not diverge. The correct statement is that the entropy cost
\emph{per unit quadratic coherence} diverges:
\[
\frac{D(\rho\|\Pi\rho)}{\|B\|_F^2}\gtrsim\log\frac{a_0}{\varepsilon_Q}\to+\infty.
\]
The logarithmic bound was previously established in
Ref.~\cite{Nasreddine2026EntropyModuli}, Lemma~5.7, Eq.~(18).
It is asymptotically sharp on the two-level family
(Proposition~\ref{prop:sharpness}).

\paragraph*{Result D. Spectral channel sensitivity}

The BKM bound decomposes as
\[
\operatorname{Tr}[B^*\Omega_{A,C}^{-1}(B)]
=\|B\|_F^2\sum_{\alpha,\beta}w_{\alpha\beta}L(a_\alpha,c_\beta),
\]
where $w_{\alpha\beta}=|\langle e_\alpha,Bf_\beta\rangle|^2/\|B\|_F^2$ and
$L(a,c)=\log(a/c)/(a-c)$ is the BKM logarithmic mean kernel
(Proposition~\ref{prop:channel-weight}). The full operator bound detects
which $(A,C)$-eigenbasis channels are active; coarse scalar data $(a_0,\varepsilon_Q,c)$
do not determine the BKM cost. The operator--scalar separation theorem
(Proposition~\ref{prop:separation}) shows the gap can be arbitrarily large.

\subsection{Scope of the results}

The one-channel variational theorem and logarithmic boundary barrier are
finite-dimensional two-block results. The midpoint BKM coercivity mechanism extends to
general finite-dimensional $\mathbb{Z}_2$-graded automorphisms
(Theorem~\ref{thm:graded-fd}) and, under bounded strict-positivity assumptions, to
finite von Neumann algebras equipped with a faithful normal finite trace (Theorem~\ref{thm:graded-semifinite}).
We do not prove a full modified logarithmic Sobolev inequality, do not treat
non-uniform dephasing rates, do not establish a theory for unbounded densities or
type~III von Neumann algebras (these require additional modular-theoretic arguments),
and do not give a complete equality classification for the variational minimizers.
A one-channel variational theorem for general $\mathbb{Z}_2$-gradings is not proved.
Section~\ref{sec:discussion} returns to these points as open directions.

\subsection{Coherence entropy as a leakage cost}

Let $P$ represent a protected, metastable, or low-energy sector, and $Q=I-P$
represent leakage modes. The block $C=Q\rho Q$ measures leaked population
($\varepsilon_Q=\operatorname{Tr}C$), while $B=P\rho Q$ measures coherent leakage. Under a
uniform floor $P\rho P\ge a_0P$, the boundary estimate implies that coherence
with a low-population sector has entropy cost per unit quadratic coherence
amplified by $\log(a_0/\varepsilon_Q)$. Since positivity forces $\|B\|_F^2$ to
vanish with $\varepsilon_Q$, the total entropy need not diverge.

\textit{Uniform-floor assumption.}
The logarithmic estimate of Corollary~\ref{cor:log-bound},
\[
D(\rho\|\Pi\rho)\ge\|B\|_F^2\log(a_0/\varepsilon_Q),
\]
is established under the uniform-floor assumption $A\ge a_0P$.
The operator BKM bound
\[
D(\rho\|\Pi\rho)\ge\operatorname{Tr}[B^*\Omega_{A,C}^{-1}(B)]
\]
retains finer spectral information: it detects which $(A,C)$-eigenbasis
channels carry coherence and weights them individually by the BKM kernel
$L(a_\alpha,c_\beta)$.

This is an information-theoretic leakage estimate. It is not a full model of
error correction, Zeno dynamics, or thermodynamic irreversibility.

\subsection{Relation to known inequalities}

\paragraph{Pinsker inequality.}
Pinsker gives $D(\rho\|\Pi\rho)\ge 2\|B\|_1^2$ (using
$\|\rho-\Pi\rho\|_1=2\|B\|_1$). The logarithmic bound dominates Pinsker
when $\log(a_0/\varepsilon_Q)\ge2\|B\|_1^2/\|B\|_F^2$. Since
$\|B\|_1^2/\|B\|_F^2\le\rank(B)$, the logarithmic bound dominates whenever
$\varepsilon_Q\le a_0\exp(-2\rank(B))$: the exponentially small leakage regime.
In this regime, the BKM boundary estimate provides a logarithmically stronger
lower bound.

\paragraph{Relative entropy of coherence.}
The quantity $D(\rho\|\Pi\rho)$ is known as the block-relative entropy of
coherence~\cite{Aberg2006,BaumgratzCramerPlenio2014,WinterYang2016,StreltsovAdessoPlenio2017,ChitambarGour2019}.
The novelty in the present paper is not the measure but: (i) the midpoint
BKM coercivity via the block-sign involution; (ii) the exact constrained
minimization identifying a one-channel optimizer.

\paragraph{BKM and Petz metrics.}
The BKM metric is the Hessian of quantum relative entropy and one of the standard
monotone metrics associated with the Bogoliubov--Kubo--Mori construction and the
Petz classification~\cite{Bogoliubov1958,Kubo1957,Mori1965,HiaiPetz1991,Petz1986,PetzRag1996,LesniewskiRuskai1999,Bhatia2007}.
The new ingredient here is the midpoint monotonicity induced by the block-sign
involution $U=P-Q$. This is not general affine Hessian monotonicity; it is a
special property of perturbations along the off-block direction.

\paragraph{Conditional expectations.}
The one-channel variational theorem, the previously known logarithmic boundary barrier, and the spectral
channel formula apply to the two-block pinching conditional expectation
$\Pi(\rho)=P\rho P\oplus Q\rho Q$. The midpoint BKM coercivity mechanism extends to
general finite-dimensional $\mathbb{Z}_2$-graded automorphisms (Section~\ref{sec:z2});
the variational results are restricted to the two-block setting.

\paragraph{Operator versus scalar bounds.}
The exact BKM channel formula (Result~D) is strictly finer than scalar estimates
based on $(a_0,\varepsilon_Q,c)$ alone, and can exceed them by arbitrarily large factors
(Proposition~\ref{prop:separation}).

\medskip
\noindent\textit{Novelty.}\ \looseness=-1
The main contribution of the present paper is
the combination, in a single two-block framework, of two \emph{independent}
ingredients: (i)~a \emph{sign-symmetric midpoint BKM coercivity} theorem, which
exploits the block-sign involution $U=P-Q$ and the scalar resolvent inequality to
establish $H_{M+tY}(Y,Y)\ge H_M(Y,Y)$ for all $t\in[0,1)$; and (ii)~an
\emph{exact constrained entropy minimization theorem} identifying a one-channel
extremizer—a single active two-level block—that achieves the minimum
block-coherence entropy at prescribed support floor, leakage mass, and
Hilbert--Schmidt off-block coherence. These two results are logically independent:
the variational theorem is not a consequence of the BKM lower bound. It uses an
independent completely positive trace-preserving (CPTP) coherence-compression construction (Theorem~\ref{thm:compression})
that directly identifies the one-channel optimizer for the constrained
entropy-geometric problem. Together, the analytic coercivity mechanism and the extremal geometry give
two complementary views of the same block-structured problem.

\paragraph{Relation to extremal entropy problems.}
Two-level reductions occur frequently in sharp quantum entropy inequalities and related trace-convexity arguments~\cite{LiebRuskai1973,Ruskai2002,OhyaPetz1993}, for example
in continuity bounds and pinching-type extremal problems. The present result differs
in two respects: the constraint fixes the Hilbert--Schmidt size of the off-block coherence
rather than a trace distance or spectrum, and a minimizer is obtained through an
explicit CPTP compression of many coherent channels into one. Thus the one-channel theorem is more specific than a convexity consequence
or a standard resource-theoretic interpretation of $D(\rho\|\Pi\rho)$~\cite{Aberg2006,MarvianSpekkens2014,ChitambarGour2019}: it is an
extremal statement for a particular block-geometric constraint set. To the author's knowledge, no existing exact extremal theorem identifies the minimum of $D(\rho\|\Pi\rho)$ under the simultaneous constraints of fixed support floor, fixed leakage mass, and fixed Hilbert--Schmidt off-block coherence. Standard coherence-resource and pinching results identify the measure and establish general monotonicity properties, while Petz equality theory characterizes equality in data processing; they do not by themselves solve this particular constrained Hilbert--Schmidt off-block coherence minimization problem.

\medskip
\noindent\textbf{Main contributions.}
\begin{enumerate}[(1)]
\item \textit{Sharp entropy-cost theorem (finite-dimensional two-block).}
  Among states with fixed block weights and fixed Hilbert--Schmidt coherence, the
  minimum block-coherence entropy can be attained by a single active two-level block.
  The high-dimensional problem collapses to a scalar $2\times2$ variational problem
  with explicit minimizer.
\item \textit{BKM midpoint coercivity and logarithmic specialization.}
  The block-sign involution gives an explicit operator-kernel lower bound.
  The known uniform-floor logarithmic estimate is recovered as a scalar specialization.

\item \textit{Auxiliary extension of the midpoint mechanism.}
  The BKM midpoint coercivity argument, but not the one-channel variational theorem,
  extends from two-block matrices to finite-dimensional trace-preserving
  $\mathbb{Z}_2$-graded automorphisms and, under bounded strict-positivity assumptions,
  to finite von Neumann algebras endowed with faithful normal finite traces
  (including finite factors and invariant finite-trace corners).
\end{enumerate}

\subsection{Paper structure}

Section~\ref{sec:setup} sets up the block geometry and BKM kernel.
Section~\ref{sec:BKM-geometry} proves the midpoint BKM coercivity, the
logarithmic boundary law, and the spectral channel formula.
Section~\ref{sec:variational} solves the one-channel variational problem.
Section~\ref{sec:dynamics} derives the dephasing dissipation corollary from the BKM midpoint bound.
Section~\ref{sec:application} gives a finite-dimensional information-theoretic
application to coherent leakage from a protected subspace. Section~\ref{sec:z2}
extends the midpoint BKM mechanism to $\mathbb{Z}_2$-graded settings: the
finite-dimensional graded theorem (Theorem~\ref{thm:graded-fd}), recovery of the
two-block bound, a finite-von-Neumann-algebra extension
(Theorem~\ref{thm:graded-semifinite}), and the resulting open directions.
Section~\ref{sec:discussion} concludes.
Appendices contain the singular-value decomposition (SVD) pinching construction, the polygon modulus lemma,
boundary cases for the Coherence Compression Principle, and the Petz-metric extension.

\section{Block Decomposition and Geometric Setup}
\label{sec:setup}

This section fixes the two-block notation and pinching geometry, introduces the
BKM operator kernel and the two-level entropy functional, and records the
continuity conventions needed for boundary-state optimizers.

\paragraph{Notation table.}
Table~\ref{tab:notation} summarizes the notation used throughout the two-block analysis.
\begin{table}[htbp]
\caption{Notation for the two-block entropy problem.}\label{tab:notation}
\renewcommand{\arraystretch}{1.3}
\begin{tabular}{ll}
\hline
Symbol & Meaning \\
\hline
$P$, $Q=I-P$ & orthogonal projections; $d_P=\dim P\cH$, $d_Q=\dim Q\cH$ \\
$A=P\rho P$, $B=P\rho Q$, $C=Q\rho Q$ & block decomposition of $\rho$ \\
$M=\Pi\rho=A\oplus C$ & block-diagonal pinching \\
$Y=\rho-M=\bigl(\begin{smallmatrix}0&B\\B^*&0\end{smallmatrix}\bigr)$ & off-diagonal coherence part \\
$\varepsilon_Q=\operatorname{Tr}(C)$ & leakage mass \\
$c=\|B\|_F^2$ & quadratic coherence \\
$a_0$ & support floor: $P\rho P\ge a_0P$ \\
$a_*=1-\varepsilon_Q-(d_P-1)a_0$ & active block diagonal entry \\
$L(a,c)=\log(a/c)/(a-c)$ & BKM logarithmic mean kernel \\
$\Phi(a,\varepsilon,x)$ & two-level relative entropy functional \\
$\Omega_{A,C}^{-1}(B)$ & BKM operator kernel \\
\hline
\end{tabular}
\end{table}

Let $\cH$ be a finite-dimensional Hilbert space with $d=\dim\cH$, $P$ an
orthogonal projection, $Q=\Id-P$.

Any density matrix $\rho$ admits the block form
\begin{equation}\label{eq:block}
\rho=\begin{pmatrix}A&B\\B^*&C\end{pmatrix},\qquad
A=P\rho P,\quad B=P\rho Q,\quad C=Q\rho Q.
\end{equation}
The block-diagonal pinching is $\Pi\rho:=A\oplus C$.
The quantity $D(\rho\|\Pi\rho)$ is the relative entropy of coherence in the
$P\oplus Q$ basis~\cite{Umegaki1962,OhyaPetz1993,BaumgratzCramerPlenio2014}.

\paragraph{Pythagorean identity.}
For any block-diagonal $\sigma$,
\begin{equation}\label{eq:pythagorean}
D(\rho\|\sigma)=D(\rho\|\Pi\rho)+D(\Pi\rho\|\sigma).
\end{equation}
In particular $D(\rho\|\sigma)\ge D(\rho\|\Pi\rho)$.

\paragraph{Strict interior versus boundary states.}
The BKM operator theorem (Theorem~\ref{thm:BKM-lower}) is proved for $A,C>0$.
The variational theorem (Theorem~\ref{thm:variational}) naturally has boundary
optimizers, because unused $Q$-directions can carry zero mass. In finite
dimensions, the scalar functional $\Phi(a,\varepsilon,x)$ is continuous on the
closed positivity domain $\{a,\varepsilon\ge0,\,0\le x\le a\varepsilon\}$, so the
variational theorem extends to boundary states by continuity (see
Lemma~\ref{lem:Phi-continuity}). We do not claim the full operator BKM theorem
for arbitrary singular $A,C$.

\begin{lemma}[Continuity of $\Phi$ at the boundary]
\label{lem:Phi-continuity}
The function
\[
\Phi(a,\varepsilon,x):=D\!\left(\begin{pmatrix}a&\sqrt{x}\\\sqrt{x}&\varepsilon\end{pmatrix}
\,\bigg\|\,\begin{pmatrix}a&0\\0&\varepsilon\end{pmatrix}\right)
\]
extends continuously to $\{a,\varepsilon\ge0,\,0\le x\le a\varepsilon\}$ with the
convention $0\log0=0$.
\end{lemma}
\begin{proof}
The eigenvalues $\lambda_\pm=(a+\varepsilon\pm\sqrt{(a-\varepsilon)^2+4x})/2$ are
continuous functions of $(a,\varepsilon,x)$ on the closed domain. The entropy
$\lambda_+\log\lambda_++\lambda_-\log\lambda_--a\log a-\varepsilon\log\varepsilon$
is continuous with the convention $0\log0=0$.
\end{proof}

\paragraph{SVD of $B$.}
Write $B=\sum_{j=1}^m s_j|u_j\rangle\langle v_j|$ with $u_j\in P\cH$,
$v_j\in Q\cH$, $s_j>0$, and set
\begin{equation}\label{eq:SVD-data}
a_j:=\langle u_j,Au_j\rangle,\qquad c_j:=\langle v_j,Cv_j\rangle.
\end{equation}
Positivity of $\rho$ gives $s_j^2\le a_jc_j$: indeed, the compression of $\rho$ to $\operatorname{span}\{u_j,v_j\}$ is the positive $2\times2$ matrix $\bigl(\begin{smallmatrix}a_j&s_j\\ s_j&c_j\end{smallmatrix}\bigr)$.

\paragraph{BKM kernel.}
For $A,C>0$, define
\[
\Omega_{A,C}^{-1}(B):=\int_0^\infty(A+rI_P)^{-1}B(C+rI_Q)^{-1}\,\dd r.
\]
The integral formula $\int_0^\infty[(a+r)(c+r)]^{-1}\,\dd r=L(a,c)$ with
\begin{equation}\label{eq:L-def}
L(a,c):=\frac{\log(a/c)}{a-c},\qquad L(a,a):=\frac{1}{a},
\end{equation}
gives $\operatorname{Tr}[B^*\Omega_{A,C}^{-1}(B)]=\sum_{\alpha,\beta}
|\langle e_\alpha,Bf_\beta\rangle|^2L(a_\alpha,c_\beta)$
in the joint eigenbases of $A$ and $C$.

\paragraph{Two-level entropy functional.}
For $a,\varepsilon\ge0$ and $0\le x\le a\varepsilon$, set
\begin{equation}\label{eq:Phi-def}
\Phi(a,\varepsilon,x):=D\!\left(
  \begin{pmatrix}a&\sqrt{x}\\\sqrt{x}&\varepsilon\end{pmatrix}
  \,\bigg\|\,
  \begin{pmatrix}a&0\\0&\varepsilon\end{pmatrix}
\right).
\end{equation}
For $a,\varepsilon>0$, this equals
$\lambda_+\log\lambda_++\lambda_-\log\lambda_--a\log a-\varepsilon\log\varepsilon$
where
\[
\lambda_\pm=\frac{a+\varepsilon\pm\sqrt{(a-\varepsilon)^2+4x}}{2}.
\]

\section{Midpoint BKM Geometry and Kernel Bounds}
\label{sec:BKM-geometry}

We derive the midpoint BKM lower bound and its resolvent operator kernel,
then recover the known scalar logarithmic specialization. The final part
expresses the bound in spectral coherence channels and compares its
information with estimates based only on scalar block parameters.

\subsection{Midpoint BKM coercivity}

\begin{lemma}[Midpoint BKM monotonicity]
\label{lem:midpoint-mono}
Let $M=A\oplus C>0$ be block-diagonal and $Y=\bigl(\begin{smallmatrix}0&B\\B^*&0\end{smallmatrix}\bigr)$
be off-diagonal with $M\pm Y\ge0$. Then:
\begin{enumerate}[(i)]
  \item $H_{M+tY}(Y,Y)=H_{M-tY}(Y,Y)$ for all $t\in[0,1)$.
  \item $H_{M+tY}(Y,Y)\ge H_M(Y,Y)$ for all $t\in[0,1)$.
\end{enumerate}
\end{lemma}

\begin{proof}
\emph{Part (i).}
Let $U:=P-Q$. Since $M=A\oplus C$ is block-diagonal and $Y$ is off-diagonal,
\[
UMU^*=M,\qquad UYU^*=-Y.
\]
Hence $M-tY=U(M+tY)U^*$. By unitary covariance of the BKM quadratic form,
\[
H_{M-tY}(Y,Y)
=H_{U(M+tY)U^*}(UYU^*,UYU^*)
=H_{M+tY}(Y,Y).
\]

\emph{Part (ii).}
Fix $r>0$ and set $R:=M+rI$, $K:=R^{-1/2}YR^{-1/2}$.
Then
\[
M\pm tY+rI=R^{1/2}(I\pm tK)R^{1/2},
\qquad
(M\pm tY+rI)^{-1}=R^{-1/2}(I\pm tK)^{-1}R^{-1/2}.
\]
Since $M>0$ and $M\pm Y\ge0$, for $0\le t<1$ we have
$M\pm tY=(1-t)M+t(M\pm Y)\ge0$. Hence $I\pm tK$ is positive and invertible,
and the spectrum of $tK$ lies in $(-1,1)$. The resolvent integrand is
\[
\begin{aligned}
&\operatorname{Tr}\!\left[Y(M\pm tY+rI)^{-1}Y(M\pm tY+rI)^{-1}\right]\\
&\quad=\operatorname{Tr}\!\left[R^{-1/2}Y R^{-1/2}(I\pm tK)^{-1}
R^{-1/2}Y R^{-1/2}(I\pm tK)^{-1}\right]\\
&\quad=\operatorname{Tr}\!\left[K(I\pm tK)^{-1}K(I\pm tK)^{-1}\right]
=\operatorname{Tr}\!\left[K^2(I\pm tK)^{-2}\right],
\end{aligned}
\]
where the last equality uses that $K$ commutes with every Borel function of $K$.
By functional calculus the pointwise scalar inequality
\[
\tfrac{1}{2}\bigl((1+u)^{-2}+(1-u)^{-2}\bigr)\ge1\quad(|u|<1)
\]
(equivalent to $u^2(3-u^2)\ge0$) gives
\[
\tfrac{1}{2}\operatorname{Tr}\!\left[K^2(I+tK)^{-2}+K^2(I-tK)^{-2}\right]
\ge\operatorname{Tr}[K^2].
\]
Integrating over $r$ and using part (i):
\[
H_{M+tY}(Y,Y)=\tfrac{1}{2}\bigl(H_{M+tY}(Y,Y)+H_{M-tY}(Y,Y)\bigr)\ge H_M(Y,Y).
\qedhere
\]
\end{proof}

\paragraph{Convexity gives part (ii) from part (i).}
For the fixed self-adjoint direction $Y$, the map
$t\mapsto H_{M+tY}(Y,Y)$ is convex on the admissible affine segment.
This is in the trace-convexity setting of
Refs.~\cite{Lieb1973,LesniewskiRuskai1999}, but the required
one-dimensional statement follows directly from the resolvent formula.
Indeed, for each $r>0$, put $R=M+rI$ and $K=R^{-1/2}YR^{-1/2}$. Along the
admissible affine segment the resolvent integrand is
$f_r(t)=\operatorname{Tr}[K^2(I+tK)^{-2}]$, with
\[
f_r''(t)=6\operatorname{Tr}[K^4(I+tK)^{-4}]\ge0.
\]
Integrating in $r$ gives convexity of $t\mapsto H_{M+tY}(Y,Y)$.
By part~(i) this map is even; hence
\[
H_M(Y,Y)\le\tfrac12\bigl(H_{M+tY}(Y,Y)+H_{M-tY}(Y,Y)\bigr)
=H_{M+tY}(Y,Y).
\]
This is the concise convexity proof of part~(ii); the preceding proof
also supplies the explicit midpoint-resolvent inequality.

\begin{remark}
The argument uses only the block-sign symmetry and the resolvent formula for
the BKM form. No monotonicity of the BKM Hessian along general affine
directions is invoked.
\end{remark}

\begin{remark}[A related midpoint estimate for Petz metrics]
\label{rem:petz-midpoint}
Lemma~\ref{lem:midpoint-mono} admits an analogue for all Petz monotone metrics:
under the same bounded strict-positivity assumptions,
$g^f_{M+tY}(Y,Y)\ge g^f_M(Y,Y)$ for $t\in[0,1)$,
proved via joint concavity of Kubo--Ando means~\cite{KuboAndo1980,Bhatia2007};
see Appendix~\ref{app:petz}.
The BKM metric is distinguished because it is the Hessian of $N\mapsto\operatorname{Tr}[N\log N]$,
allowing the midpoint estimate to integrate directly into the entropy lower bound.
\end{remark}

\begin{theorem}[BKM operator lower bound]
\label{thm:BKM-lower}
Let $\rho=\bigl(\begin{smallmatrix}A&B\\B^*&C\end{smallmatrix}\bigr)$ be a
density matrix with $A,C>0$. Then
\begin{equation}\label{eq:BKM-intrinsic}
D(\rho\|\Pi\rho)\;\ge\;\operatorname{Tr}\!\left[B^*\,\Omega_{A,C}^{-1}(B)\right].
\end{equation}
\end{theorem}

\begin{proof}
Set $M:=\Pi\rho=A\oplus C$ and $Y:=\rho-\Pi\rho$.
Define $F(t):=D(M+tY\|M)$ for $0\le t<1$.
Since $A,C>0$ we have $M>0$, hence $M+tY=(1-t)M+t\rho>0$ for $0\le t<1$,
so $F$ is twice differentiable on $[0,1)$ and continuous at $t=1$.

We have $F(0)=0$. Moreover,
\[
F'(t)=\operatorname{Tr}\!\left[
Y\bigl(\log(M+tY)+I-\log M\bigr)
\right].
\]
At $t=0$ this gives
\[
F'(0)=\operatorname{Tr}\!\left[Y(\log M+I-\log M)\right]
       =\operatorname{Tr}(Y)=0.
\]
Equivalently, the possible term $\operatorname{Tr}(Y\log M)$ vanishes because
$\log M$ is block-diagonal while $Y$ is off-diagonal, and $\operatorname{Tr}(Y)=0$
since $\rho$ and $\Pi\rho$ have the same trace.
The second derivative is $F''(t)=H_{M+tY}(Y,Y)$.
By Lemma~\ref{lem:midpoint-mono}, $H_{M+tY}(Y,Y)\ge H_M(Y,Y)$. Taylor's
integral formula gives $F(t)=\int_0^t(t-s)F''(s)\,ds\ge\frac{t^2}{2}H_M(Y,Y)$.
Letting $t\uparrow1$ yields $D(\rho\|\Pi\rho)\ge\frac{1}{2}H_M(Y,Y)$.

Using the block structure $(M+rI)^{-1}=(A+rI)^{-1}\oplus(C+rI)^{-1}$,
\[
H_M(Y,Y)
=2\int_0^\infty\operatorname{Tr}\!\left[B^*(A+rI)^{-1}B(C+rI)^{-1}\right]\dd r
=2\,\operatorname{Tr}\!\left[B^*\Omega_{A,C}^{-1}(B)\right].
\]
Hence $D(\rho\|\Pi\rho)\ge\operatorname{Tr}[B^*\Omega_{A,C}^{-1}(B)]$.
\end{proof}

\begin{remark}[Singular boundary states]
\label{rem:boundary}
For states with $A,C>0$, Theorem~\ref{thm:BKM-lower} applies directly and no
regularization is involved. If $A$ or $C$ is singular, one may apply the theorem
to the regularized states $\rho_\delta=(1-\delta)\rho+(\delta/d)I$. A full
passage of the operator-kernel bound to the singular limit is a separate
convergence question and is not used in the sequel.
\end{remark}

\subsection{The logarithmic boundary law}

\paragraph{Boundary convention.}
The operator BKM estimate (Theorem~\ref{thm:BKM-lower}) is stated in the strict
regime $A,C>0$. The scalar variational quantities $\Phi(a,\varepsilon,x)$ extend
continuously to the closed positivity domain $a,\varepsilon\ge0$, $0\le x\le a\varepsilon$,
with the convention $0\log0=0$ (Lemma~\ref{lem:Phi-continuity}). In the
logarithmic boundary estimates below we state the nontrivial case $\varepsilon_Q>0$.
If $\varepsilon_Q=0$, then positivity of $\rho$ forces $B=0$, and the displayed
logarithmic inequality holds in the limiting sense with both sides equal to zero.
The strict BKM operator theorem is for $A,C>0$; the variational theorem is scalar
after the SVD pinching and uses the continuous extension of $\Phi$, so its optimizer
may lie on the boundary.

\begin{corollary}[Boundary logarithmic specialization]
\label{cor:log-bound}
Assume $A\ge a_0P$ and $0<\varepsilon_Q:=\operatorname{Tr}(C)\le a_0/2$.
Then
\begin{equation}\label{eq:log-bound}
D(\rho\|\Pi\rho)\;\ge\;\|B\|_F^2\,\log\!\left(\frac{a_0}{\varepsilon_Q}\right).
\end{equation}
For any block-diagonal $\sigma$, $D(\rho\|\sigma)\ge D(\rho\|\Pi\rho)$
by~\eqref{eq:pythagorean}, so~\eqref{eq:log-bound} holds with $D(\rho\|\sigma)$
in place of $D(\rho\|\Pi\rho)$.
If $\varepsilon_Q=0$, positivity forces $B=0$ and the inequality is interpreted by the
limiting convention above.
\end{corollary}

The logarithmic estimate is an earlier result (Ref.~\cite{Nasreddine2026EntropyModuli},
Lemma 5.7, Eq.~(18)); its formulation here places it alongside the
operator-kernel estimate, which keeps finer spectral information.

\begin{proof}
By the SVD pinching reduction (Appendix~\ref{app:svd}),
\[
D(\rho\|\Pi\rho)\ge\sum_{s_j>0}\Phi(a_j,c_j,s_j^2).
\]
By Lemma~\ref{lem:phi-bkm-lower}, $\Phi(a_j,c_j,s_j^2)\ge s_j^2\,L(a_j,c_j)$
for each active term. For each $s_j>0$, positivity of the corresponding $2\times2$
compression gives $s_j^2\le a_jc_j$, so $c_j>0$ and $L(a_j,c_j)$ is well-defined.
Since $a_j\ge a_0>0$, $0<c_j\le\varepsilon_Q\le a_0/2<a_0\le a_j$, we have
$0<a_j-c_j\le1$ (as diagonal entries of a density matrix). Therefore
\[
L(a_j,c_j)=\frac{\log(a_j/c_j)}{a_j-c_j}\ge\log(a_j/c_j)\ge\log\!\left(\frac{a_0}{\varepsilon_Q}\right).
\]
Summing over active channels and using $\sum_{s_j>0}s_j^2=\|B\|_F^2$
yields~\eqref{eq:log-bound}.
\end{proof}

\begin{proposition}[Asymptotic sharpness]
\label{prop:sharpness}
Let $\rho_q=\bigl(\begin{smallmatrix}1-q&\sqrt{q(1-q)}\\\sqrt{q(1-q)}&q\end{smallmatrix}\bigr)$
and $\Pi\rho_q=\diag(1-q,q)$ for $q\in(0,\tfrac{1}{2})$.
Setting $A_q=1-q$, $C_q=q$, $B_q=\sqrt{q(1-q)}$:

\noindent\emph{(a) BKM asymptotic saturation.} As $q\to0^+$,
\begin{equation}\label{eq:asymp-exact}
D(\rho_q\|\Pi\rho_q)=\operatorname{Tr}[B_q^*\Omega_{A_q,C_q}^{-1}(B_q)]+O(q),
\quad\text{hence}\quad
\frac{D(\rho_q\|\Pi\rho_q)}{\operatorname{Tr}[B_q^*\Omega_{A_q,C_q}^{-1}(B_q)]}\to1.
\end{equation}

\noindent\emph{(b) Logarithmic sharpness.} Fix any $a_0\in(0,1)$
independently of $q$; the floor condition $A_q=1-q\ge a_0$ then holds for
all sufficiently small $q$. As $q\to0^+$,
\begin{equation}\label{eq:log-sharp}
D(\rho_q\|\Pi\rho_q)\sim\|B_q\|_F^2\log\!\left(\frac{1-q}{q}\right)
\asymp q\log(1/q),
\end{equation}
so $D(\rho_q\|\Pi\rho_q)/[\|B_q\|_F^2\log(a_0/\varepsilon_Q(\rho_q))]\to1$.
\end{proposition}

\begin{proof}
Since $\rho_q$ is pure, $D(\rho_q\|\Pi\rho_q)=h(q)=q\log(1/q)+O(q)$.
The BKM form is
\[
\operatorname{Tr}[B_q^*\Omega_{A_q,C_q}^{-1}(B_q)]=q(1-q)L(1-q,q)
\]
where $L(a,c)=\log(a/c)/(a-c)$. Since $(1-q)/(1-2q)=1+O(q)$ and $\log(1-q)=O(q)$,
$q(1-q)L(1-q,q)=q\log(1/q)+O(q)$, giving~\eqref{eq:asymp-exact}.
Since $\|B_q\|_F^2=q(1-q)\sim q$ and $\log((1-q)/q)\sim\log(1/q)$,
part~(b) follows. Together: the BKM operator controls the exact leading asymptotic
term; the logarithmic corollary therefore captures the same leading order.
\end{proof}

\subsection{Spectral channel formula and operator--scalar separation}
\label{ssec:spectral}

\begin{proposition}[Channel-weighted representation]
\label{prop:channel-weight}
Assume $A,C>0$ and $B\ne0$. Let $A=\sum_\alpha a_\alpha|e_\alpha\rangle\langle e_\alpha|$ and
$C=\sum_\beta c_\beta|f_\beta\rangle\langle f_\beta|$ be the spectral decompositions
(so $a_\alpha,c_\beta>0$).
Define the coherence channel weights
\[
w_{\alpha\beta}:=\frac{|\langle e_\alpha,Bf_\beta\rangle|^2}{\|B\|_F^2}\ge0,
\quad\sum_{\alpha,\beta}w_{\alpha\beta}=1.
\]
Then
\[
\frac{\operatorname{Tr}[B^*\Omega_{A,C}^{-1}(B)]}{\|B\|_F^2}
=\sum_{\alpha,\beta}w_{\alpha\beta}\,L(a_\alpha,c_\beta),
\]
and in particular $\operatorname{Tr}[B^*\Omega_{A,C}^{-1}(B)]$ is well-defined and finite.
\end{proposition}

\begin{proof}
Direct substitution of the eigendecompositions of $A$ and $C$ into
$\Omega_{A,C}^{-1}(B)=\int_0^\infty(A+rI)^{-1}B(C+rI)^{-1}\,dr$, followed by
$\int_0^\infty(a_\alpha+r)^{-1}(c_\beta+r)^{-1}\,dr=L(a_\alpha,c_\beta)$.
\end{proof}

The scalar law $\log(a_0/\varepsilon_Q)$ is a universal lower bound on each
channel BKM kernel value: since $a_\alpha-c_\beta\le1$ for density-matrix blocks,
$L(a_\alpha,c_\beta)\ge\log(a_\alpha/c_\beta)\ge\log(a_0/\varepsilon_Q)$.
However, this scalar lower bound may be far from the actual average.

\begin{remark}[Scalar law versus operator bound]
\label{rem:coarse}
The scalar logarithmic law follows from the operator bound by two steps:
the channel-weighted representation and the pointwise estimate $L(a,c)\ge\log(a/c)$
(valid when $a-c\le1$). The operator bound detects the active channels;
the coarse parameters $(a_0,\varepsilon_Q,c)$ do not.

The floor-based logarithmic estimate depends only on $a_0$,
$\varepsilon_Q$, and $\|B\|_F$. However, it can be
arbitrarily pessimistic when the coherence avoids the spectral channel associated
with the smallest eigenvalue of $A$. The BKM operator bound detects the active
channels and therefore contains strictly finer information.
\end{remark}

\begin{proposition}[Operator--scalar separation]
\label{prop:separation}
For every $K>1$ there exists a finite-dimensional density matrix
$\rho=\bigl(\begin{smallmatrix}A&B\\B^*&C\end{smallmatrix}\bigr)$ such that
\[
\frac{\operatorname{Tr}[B^*\Omega_{A,C}^{-1}(B)]}{\|B\|_F^2\log(\lambda_{\min}(A)/\operatorname{Tr}C)}
\;\ge\;K.
\]
\end{proposition}

\begin{proof}
Set $\eta=1/(K+1)\in(0,\tfrac{1}{2})$ and choose $\varepsilon\in(0,\tfrac{1}{4})$ small.
Define $m=\varepsilon^{1-\eta}$ (so $m>\varepsilon$ since $1-\eta<1$),
$a_1=1-m-\varepsilon$ (near~$1$ for small $\varepsilon$), and the density matrix with
$A=\diag(a_1,m)$, $C=(\varepsilon)$, $B=(\sqrt{a_1\varepsilon/2},0)^T$.
The Schur complement of $A$ in the block matrix is
$\varepsilon - B^*A^{-1}B = \varepsilon - (a_1\varepsilon/2)/a_1 = \varepsilon/2 > 0$,
so the block matrix is positive semidefinite. Hence $\rho$ is a valid density matrix.
$\lambda_{\min}(A)=m$,
$\operatorname{Tr}C=\varepsilon$. Since $m=\varepsilon^{1-\eta}$: $\log(m/\varepsilon)=\eta\log(1/\varepsilon)$.
The BKM form concentrates on channel $(a_1,\varepsilon)$ with $a_1\to1$, giving
\[
\frac{\operatorname{Tr}[B^*\Omega_{A,C}^{-1}(B)]}{\|B\|_F^2\log(m/\varepsilon)}
=\frac{L(a_1,\varepsilon)}{\log(m/\varepsilon)}\to\frac{1}{\eta}=K+1>K
\quad\text{as }\varepsilon\to0.
\]
For all sufficiently small $\varepsilon>0$, the ratio exceeds $K$. The separation
arises because the scalar law is anchored to the inactive eigenvalue $m\approx\varepsilon^{1-\eta}$,
while the BKM form depends on the active channel $a_1\approx1$.
\end{proof}

\subsection{Comparison with Pinsker's inequality}

Pinsker applied to $\sigma=\Pi\rho$ gives
$D(\rho\|\Pi\rho)\ge\tfrac{1}{2}\|\rho-\Pi\rho\|_1^2=2\|B\|_1^2$.
The logarithmic bound dominates Pinsker when $\log(a_0/\varepsilon_Q)\ge2\|B\|_1^2/\|B\|_F^2$.
Since $\|B\|_1^2/\|B\|_F^2\le\rank(B)$, this holds whenever
\[
\varepsilon_Q\le a_0\exp(-2\rank(B)),
\]
i.e., in the exponentially small leakage regime.

\section{One-Channel Extremizers and Variational Problem}
\label{sec:variational}

The previous sections give lower bounds on $D(\rho\|\Pi\rho)$. We now solve
the associated extremal problem: fixing the protected-sector floor, the leakage
mass, and the Hilbert--Schmidt size of the off-block coherence, what is the
least possible block-coherence entropy?

The answer is given by the \emph{one-channel principle}: the infimum can be
attained by concentrating the prescribed coherence into one effective two-level
channel. Thus the high-dimensional minimization problem collapses to a scalar
$2\times2$ problem.
The technical mechanism is the Coherence Compression Principle (Theorem~\ref{thm:compression}):
any state with multiple coherent blocks can be mapped by a CPTP compression to one
with a single active block, without increasing entropy.

\paragraph{Why this class?}
The admissibility conditions are not an auxiliary technical restriction. They encode
the minimal data needed to formulate the entropy-cost problem: the class fixes the
protected-sector energy distribution (floor $a_0$), the total leaked population
$\varepsilon_Q$, and the Hilbert--Schmidt size of the coherent leakage $c$, while
allowing all other degrees of freedom (coherence directions, phases, spectator
weights) to vary. The Hilbert--Schmidt constraint is the natural quadratic datum
compatible with the BKM Hessian at the block-diagonal point; it is also the
quantity directly decomposed by the singular-value pinching used in the proof.
The conditions $1-\varepsilon_Q\ge d_Pa_0$ and $c\le a_*\varepsilon_Q$
are precisely the Schur-complement positivity conditions for the active $2\times2$
block of the displayed optimizer; they are the hypotheses under which the theorem
constructs an attaining state.

\textit{Example.} A qubit on $\mathcal{H}=\mathbb{C}^2$ with $P=|0\rangle\langle0|$,
$Q=|1\rangle\langle1|$, $d_P=d_Q=1$: the admissible class is all $\rho$ with
$\langle0|\rho|0\rangle=a_*$, $\langle1|\rho|1\rangle=\varepsilon_Q$, and
$|\langle0|\rho|1\rangle|^2=c$, i.e.\ all qubit states with fixed populations and
fixed off-diagonal modulus. The optimizer is the unique such state (up to phase).

\textit{Non-example.} Without the coherence constraint $c$ fixed, the infimum is
$D=0$ (achieved at $B=0$), which is trivially minimized by the block-diagonal state;
no interesting entropy cost is captured. Fixing $c>0$ forces the state to carry
genuine coherent leakage, and the theorem quantifies exactly what this costs.

The proof pipeline has three steps:
\begin{enumerate}[(1)]
\item \emph{SVD pinching}: apply the CPTP channel $\mathcal{E}$ of
Appendix~\ref{app:svd} to decompose the state into orthogonal singular-vector
blocks; by the exact decomposition proved there,
$D(\rho\|\Pi\rho)\ge\sum_j\Phi(a_j^{\mathrm{pin}},c_j^{\mathrm{pin}},s_j^2)$.
\item \emph{Merging}: apply Theorem~\ref{thm:compression} to combine these scalar
problems monotonically into a single two-level problem $\Phi(A,E,X)$.
\item \emph{Two-level minimization}: solve the resulting scalar optimization
explicitly.
\end{enumerate}

\begin{remark}[SVD pinching versus exact BKM channel formula]
The exact BKM channels are spectral channels associated with the eigenbases of $(A,C)$:
\[
\operatorname{Tr}[B^*\Omega_{A,C}^{-1}(B)]
=\sum_{\alpha,\beta}|\langle e_\alpha,Bf_\beta\rangle|^2L(a_\alpha,c_\beta).
\]
The SVD-based formulas used in the variational reduction are lower-bound reductions
obtained through data processing, not the canonical spectral decomposition.
\end{remark}

\subsection{Properties of the two-level functional}

\begin{lemma}[Two-level BKM lower bound; cf.\ Ref.~\cite{Nasreddine2026EntropyModuli}, Lemma 5.5 and Corollary 5.6]
\label{lem:phi-bkm-lower}
Let $a,\varepsilon>0$ and $0\le x\le a\varepsilon$. Then
\[
\Phi(a,\varepsilon,x)\;\ge\;x\,L(a,\varepsilon).
\]
\end{lemma}
\begin{proof}
Apply Theorem~\ref{thm:BKM-lower} to the normalized $2\times2$ state
\[
\sigma:=\frac{1}{a+\varepsilon}\begin{pmatrix}a&\sqrt{x}\\\sqrt{x}&\varepsilon\end{pmatrix},
\]
with block entries $A_\sigma=a/(a+\varepsilon)$, $C_\sigma=\varepsilon/(a+\varepsilon)$,
$B_\sigma=\sqrt{x}/(a+\varepsilon)$, and diagonal pinching
$\Pi\sigma=\diag(A_\sigma,C_\sigma)$. Since $D(\sigma\|\Pi\sigma)=\Phi(a,\varepsilon,x)/(a+\varepsilon)$
(relative entropy scales as a density), Theorem~\ref{thm:BKM-lower} gives
\[
\frac{\Phi(a,\varepsilon,x)}{a+\varepsilon}
=D(\sigma\|\Pi\sigma)
\ge\operatorname{Tr}\!\bigl[B_\sigma^*\,\Omega_{A_\sigma,C_\sigma}^{-1}(B_\sigma)\bigr]
=\frac{x}{(a+\varepsilon)^2}\,L\!\left(\frac{a}{a+\varepsilon},\frac{\varepsilon}{a+\varepsilon}\right).
\]
The logarithmic mean satisfies $L(\lambda a,\lambda\varepsilon)=\lambda^{-1}L(a,\varepsilon)$
for $\lambda>0$, so $L(a/(a+\varepsilon),\varepsilon/(a+\varepsilon))=(a+\varepsilon)L(a,\varepsilon)$.
Consequently the normalized quadratic term equals
$\frac{x}{a+\varepsilon}L(a,\varepsilon)$; multiplying the inequality
by $a+\varepsilon$ gives $\Phi(a,\varepsilon,x)\ge xL(a,\varepsilon)$.
\end{proof}

\begin{proposition}[Properties of $\Phi$]
\label{prop:Phi-properties}
Let $a,\varepsilon>0$.
\begin{enumerate}[(1)]
\item $x\mapsto\Phi(a,\varepsilon,x)$ is strictly increasing and strictly convex
  on $[0,a\varepsilon)$, with $\Phi(a,\varepsilon,0)=0$.
\item $\Phi(a,\varepsilon,x)=x\cdot L(a,\varepsilon)(1+O(x))$ as $x\to0$.
\item If additionally $0<\varepsilon<a\le1$, then for all $x\in[0,a\varepsilon]$:
\begin{equation}\label{eq:Phi-lower}
\Phi(a,\varepsilon,x)\;\ge\;x\log(a/\varepsilon).
\end{equation}
\end{enumerate}
\end{proposition}

\begin{proof}
\emph{(1)} Set $\Delta:=\sqrt{(a-\varepsilon)^2+4x}>0$ for $x>0$.
Differentiating $\lambda_\pm=(a+\varepsilon\pm\Delta)/2$ gives
\[
\partial_x\Phi=\frac{\log(\lambda_+/\lambda_-)}{\Delta}>0
\]
(strict monotonicity; the numerator is positive since $\lambda_+>\lambda_-$).
At $a=\varepsilon$, $x=0$, this formula is $0/0$; by continuity, writing $s=\sqrt{x}$
the eigenvalues are $a\pm s+O(s^2)$ and $\partial_x\Phi(a,a,0)=\lim_{s\to0}(2s/(2s))=1/a=L(a,a)$.
Setting $u:=\tfrac{1}{2}\log(\lambda_+/\lambda_-)>0$ for $x>0$, strict convexity follows from
$\partial_x^2\Phi=4(\tfrac{1}{2}\sinh(2u)-u)/\Delta^3>0$
since $\sinh(v)>v$ for $v=2u>0$.
\emph{(2)} At $x=0$ (and $a\ne\varepsilon$), $\lambda_\pm=a,\varepsilon$ and
$\partial_x\Phi|_{x=0}=\log(a/\varepsilon)/(a-\varepsilon)=L(a,\varepsilon)$.
The case $a=\varepsilon$ is covered by the limit above.
\emph{(3)} By Lemma~\ref{lem:phi-bkm-lower},
$\Phi(a,\varepsilon,x)\ge xL(a,\varepsilon)\ge x\log(a/\varepsilon)$
since $0<a-\varepsilon\le1$.
\end{proof}

\begin{remark}[Domain hypothesis for part (3)]
Part (3) requires $0<\varepsilon<a\le1$, which ensures $L(a,\varepsilon)=\log(a/\varepsilon)/(a-\varepsilon)\ge\log(a/\varepsilon)$
(since $a-\varepsilon\le1$). This condition is automatically satisfied when $a,\varepsilon$ are diagonal
entries of a density matrix.
Parts (1) and (2) hold for all $a,\varepsilon>0$ without the upper bound on $a$.
\end{remark}

\subsection{The Coherence Compression Principle}

We isolate the core structural result as a theorem. Here $\delta_j\ge0$ denotes the excess of the
$j$-th protected-sector diagonal entry above the floor $a_0$.

\paragraph{Boundary conventions.}
The proof below is written for the generic case $X>0$ and for channels with nonzero coherent
weight. Terms with $x_j=0$ contribute no off-diagonal coherence and may be omitted from the
active phase-alignment step without affecting the inequality. If $\varepsilon_j=0$, positivity of
the corresponding $2\times2$ block forces $x_j=0$. If $E=0$, then all $\varepsilon_j=0$ and
$\varepsilon_{\mathrm{rem}}=0$, which by positivity forces $X=0$; this case is therefore
subsumed under $X=0$. Boundary cases ($\varepsilon_j\to0$ or $x_j\to0$) are obtained by
continuity of $\Phi$ on the closed positivity domain (Lemma~\ref{lem:Phi-continuity}).
Appendix~\ref{app:edge} records the same cases explicitly.

\begin{theorem}[Coherence Compression Principle]
\label{thm:compression}
Let $a_0>0$, $\delta_j\ge0$, $\varepsilon_j\ge0$, $\varepsilon_{\mathrm{rem}}\ge0$,
and $0\le x_j\le(a_0+\delta_j)\varepsilon_j$ for each $j$. Set $a_j:=a_0+\delta_j$ and
\[
A:=a_0+\textstyle\sum_j\delta_j,\quad
E:=\varepsilon_{\mathrm{rem}}+\textstyle\sum_j\varepsilon_j,\quad
X:=\textstyle\sum_j x_j.
\]
Assume $X\le AE$. Then
\begin{equation}\label{eq:merging}
\sum_j\Phi(a_j,\varepsilon_j,x_j)\;\ge\;\Phi(A,E,X).
\end{equation}
\end{theorem}

The theorem says: given finitely many two-level coherent blocks with individual
protected-mass excesses $\{\delta_j\}$ above floor $a_0$, leakages $\{\varepsilon_j\}$, and
coherences $\{x_j\}$, there exists a CPTP compression to a single active two-level
block with aggregated parameters $(A,E,X)$, and relative entropy to the block-diagonal
algebra cannot increase. Boundary cases are recorded in Appendix~\ref{app:edge}.

\begin{proof}
Write
\[
M=\bigoplus_j\begin{pmatrix}a_j&\sqrt{x_j}\\\sqrt{x_j}&\varepsilon_j\end{pmatrix}
\oplus[\varepsilon_{\mathrm{rem}}],\qquad
D=\bigoplus_j\begin{pmatrix}a_j&0\\0&\varepsilon_j\end{pmatrix}\oplus[\varepsilon_{\mathrm{rem}}],
\]
so $D(M\|D)=\sum_j\Phi(a_j,\varepsilon_j,x_j)$.

The operators $M$ and $D$ have the same trace $T:=\operatorname{Tr}(M)=\operatorname{Tr}(D)$,
since they have identical diagonals. By homogeneity of Umegaki relative entropy~\cite{Umegaki1962,OhyaPetz1993},
$D(M\|D)=T\,D(M/T\|D/T)$. The usual data-processing inequality for density matrices~\cite{Lindblad1975,Ruskai2002}
therefore implies the same inequality for the unnormalized pair $(M,D)$.

We construct an explicit CPTP map $\mathcal{E}$ with
$D(\mathcal{E}(M)\|\mathcal{E}(D))=\Phi(A,E,X)$;
data-processing then gives the result.

Let $p_j,q_j$ be the input basis vectors of block $j$ and $p,q$ the output active vectors.

\emph{Step 1: Feasibility.} Assume $X>0$ (the case $X=0$ is trivial). Construct
$r_j\in[0,1]$ with $\sum_j a_j r_j^2=A$ and $\sum_j r_j\sqrt{x_j}\ge\sqrt{X}$.
Set $r_j^{(0)}=\sqrt{x_j/X}\in[0,1]$. Since $a_j=a_0+\delta_j$ and
$\delta_j\ge0$, we have $a_j\le a_0+\sum_i\delta_i=A$ for every $j$. Hence
\[
\sum_j a_j(r_j^{(0)})^2
=
\frac{\sum_j a_jx_j}{X}
\le \max_j a_j
\le A.
\]
For the upper endpoint $r_j=1$, we need $\sum_{j=1}^k a_j\ge A$.
Let $k$ be the number of scalar blocks in the family (the number of indices $j$).
Indeed,
\[
\sum_{j=1}^k a_j
=
ka_0+\sum_{j=1}^k\delta_j
=
A+(k-1)a_0
\ge A,
\]
since $a_0>0$ and $k\ge1$ whenever $X>0$.
By the intermediate value theorem along the continuous path $r_j(t)=(1-t)r_j^{(0)}+t$,
$t\in[0,1]$, there exists $t_*\in[0,1]$ with $\sum_ja_jr_j(t_*)^2=A$.
Set $r_j:=r_j(t_*)$; since $r_j\ge r_j^{(0)}$,
\[
\sum_jr_j\sqrt{x_j}\ge\sum_jr_j^{(0)}\sqrt{x_j}=\frac{\sum_j\sqrt{x_j}\,\sqrt{x_j}}{\sqrt{X}}
=\frac{X}{\sqrt{X}}=\sqrt{X}.
\]

\emph{Step 2: Phase alignment.} Set $\ell_j:=r_j\sqrt{x_j}$ for each $j$.
We verify the two hypotheses of the polygon modulus lemma (Appendix~\ref{app:polygon}):
\begin{itemize}
\item $\sum_j\ell_j\ge\sqrt{X}$: established in Step~1.
\item $\max_j\ell_j\le\sqrt{X}$: since $r_j\le1$ and $x_j\le X$ (as $X=\sum_ix_i$ and $x_i\ge0$),
  we have $\ell_j=r_j\sqrt{x_j}\le\sqrt{x_j}\le\sqrt{X}$.
\end{itemize}
The polygon modulus lemma therefore gives phases $\theta_j\in\mathbb{R}$
with $|\sum_je^{i\theta_j}\ell_j|=\sqrt{X}$. After a global phase rotation on $|q\rangle$,
assume $\sum_je^{i\theta_j}r_j\sqrt{x_j}=\sqrt{X}$.
Set $\alpha_j:=e^{i\theta_j}r_j$, so $|\alpha_j|=r_j$ and
\[
(\mathcal{E}(M))_{pq}=\sum_j\alpha_j\sqrt{x_j}=\sqrt{X}.
\]

\emph{Step 3: Kraus operators.}
The input Hilbert space for $\mathcal{E}$ has orthonormal basis
$\{|p_j\rangle\}_{j=1}^k\cup\{|q_j\rangle\}_{j=1}^k\cup\{|q_{\mathrm{rem}}\rangle\}$,
where $|p_j\rangle$ and $|q_j\rangle$ are the active protected and leakage vectors of block $j$,
and $|q_{\mathrm{rem}}\rangle$ carries the residual leakage $\varepsilon_{\mathrm{rem}}$.
The output Hilbert space has active vectors $|p\rangle$, $|q\rangle$,
and orthogonal spectator vectors $\{|s_j\rangle\}_{j=1}^k$.
Define the Kraus operators:
\[
K_j^{(a)}:=\alpha_j|p\rangle\langle p_j|+|q\rangle\langle q_j|,\quad
K_j^{(s)}:=\sqrt{1-|\alpha_j|^2}\,|s_j\rangle\langle p_j|,\quad
K_{\mathrm{rem}}:=|q\rangle\langle q_{\mathrm{rem}}|.
\]
The spectator vectors $s_j$ are mutually orthogonal and orthogonal to $p$, $q$.

\emph{Step 4: Completeness.}
We verify $\sum_j\bigl[(K_j^{(a)})^*K_j^{(a)}+(K_j^{(s)})^*K_j^{(s)}\bigr]+K_{\mathrm{rem}}^*K_{\mathrm{rem}}=I$
on the input space. Computing each Kraus term:
\begin{align*}
(K_j^{(a)})^*K_j^{(a)}&=|\alpha_j|^2|p_j\rangle\langle p_j|+|q_j\rangle\langle q_j|,\\
(K_j^{(s)})^*K_j^{(s)}&=(1-|\alpha_j|^2)|p_j\rangle\langle p_j|,\\
K_{\mathrm{rem}}^*K_{\mathrm{rem}}&=|q_{\mathrm{rem}}\rangle\langle q_{\mathrm{rem}}|.
\end{align*}
Summing over $j$ and collecting:
\[
\sum_j\bigl[(K_j^{(a)})^*K_j^{(a)}+(K_j^{(s)})^*K_j^{(s)}\bigr]
=\sum_j|p_j\rangle\langle p_j|+\sum_j|q_j\rangle\langle q_j|.
\]
Adding $K_{\mathrm{rem}}^*K_{\mathrm{rem}}$ and using that
$\{|p_j\rangle\}\cup\{|q_j\rangle\}\cup\{|q_{\mathrm{rem}}\rangle\}$ is an orthonormal basis of the input space
(by the orthonormal basis property of the input space stated in Step~3)
gives the total $=I$, so $\mathcal{E}$ is CPTP.

\emph{Step 5: Output.}
We compute $\mathcal{E}(M)$ and $\mathcal{E}(D)$ in the basis $\{|p\rangle,|q\rangle,|s_j\rangle\}$ of the
output space. The only Kraus operators contributing to the $(p,p)$ output element are $K_j^{(a)}$:
\[
\langle p|\mathcal{E}(M)|p\rangle
=\sum_j\langle p|K_j^{(a)}\bigl(a_j|p_j\rangle\langle p_j|+\cdots\bigr)(K_j^{(a)})^*|p\rangle
=\sum_j|\alpha_j|^2 a_j=\sum_jr_j^2 a_j=A,
\]
where the last equality uses the constraint from Step~1. Similarly:
\[
\langle q|\mathcal{E}(M)|q\rangle
=\sum_j\varepsilon_j+\varepsilon_{\mathrm{rem}}=E,
\quad
\langle p|\mathcal{E}(M)|q\rangle
=\sum_j\alpha_j\sqrt{x_j}=\sqrt{X}\quad\text{(from Step~2)}.
\]
The spectator contribution is
$\langle s_j|\mathcal{E}(M)|s_j\rangle=(1-|\alpha_j|^2)a_j=(1-r_j^2)a_j$,
so $S:=\bigoplus_j(1-r_j^2)a_j|s_j\rangle\langle s_j|$.
Since $|p\rangle$, $|q\rangle$, and $\{|s_j\rangle\}$ are mutually orthogonal output vectors
(by construction in Step~3), $\mathcal{E}(M)$ and $\mathcal{E}(D)$ block-diagonalize as:
\[
\mathcal{E}(M)=\begin{pmatrix}A&\sqrt{X}\\\sqrt{X}&E\end{pmatrix}\oplus S,
\qquad
\mathcal{E}(D)=\begin{pmatrix}A&0\\0&E\end{pmatrix}\oplus S.
\]
(The off-diagonal $\mathcal{E}(D)_{pq}=\sum_j\alpha_j\cdot0=0$ since $D$ has zero off-diagonal blocks.)

By additivity of relative entropy over orthogonal direct summands:
\[
D(\mathcal{E}(M)\|\mathcal{E}(D))
=D\!\left(\begin{pmatrix}A&\sqrt{X}\\\sqrt{X}&E\end{pmatrix}\bigg\|\begin{pmatrix}A&0\\0&E\end{pmatrix}\right)
+D(S\|S)
=\Phi(A,E,X)+0.
\]
Data-processing then gives $D(M\|D)\ge D(\mathcal{E}(M)\|\mathcal{E}(D))=\Phi(A,E,X)$, completing the proof of~\eqref{eq:merging}.
The CPTP map is used here only as a data-processing certificate for the scalar merging inequality;
its output space is an auxiliary finite-dimensional space.
\end{proof}

\begin{remark}[Interpretation of the Coherence Compression Principle]
The theorem establishes that the constrained minimum can be achieved by a two-level
configuration. No collection of multiple coherent channels can achieve a lower
block-coherence entropy than the single active two-level channel obtained by
compressing them. The floor-aware merging CPTP map is the explicit certificate of
this compression.
\end{remark}

\subsection{The one-channel variational theorem}

The main solved extremal problem can be read as follows in entropy-geometric terms.
The block-diagonal face $\{\rho:\rho=\Pi\rho\}$ is the ``incoherent'' boundary of
the state space relative to the $P\oplus Q$ decomposition. The block-coherence entropy
$D(\rho\|\Pi\rho)$ measures the relative-entropy distance from $\rho$ to its closest
point on this face (its own projection $\Pi\rho$). The variational theorem answers the
following natural entropy-geometric extremal problem:

\begin{quote}
\emph{Among all states with fixed leakage mass $\varepsilon_Q=\operatorname{Tr}(Q\rho Q)$,
fixed quadratic off-block coherence $c=\|P\rho Q\|_F^2$, and fixed support floor
$P\rho P\ge a_0P$, which state lies closest to the block-diagonal face in relative entropy?}
\end{quote}

One answer attaining the minimum is the one-channel optimizer $\rho_*$: a two-level active block plus
incoherent spectators. The high-dimensional constrained problem collapses to a scalar
$2\times2$ problem. This collapse is not a coincidence but a structural consequence of
the coherence compression principle (Theorem~\ref{thm:compression}): the CPTP merging map
witnesses that any multi-channel distribution of coherence can be compressed to a single
channel without moving closer to the block-diagonal face.

\begin{lemma}[Feasibility of the admissible class]
\label{lem:admissible-feasibility}
Let $a_*:=1-\varepsilon_Q-(d_P-1)a_0$. The admissible class
$\mathcal{A}(a_0,\varepsilon_Q,c)$ is nonempty if and only if
\[
1-\varepsilon_Q\ge d_Pa_0,\qquad 0\le c\le a_*\varepsilon_Q.
\]
\end{lemma}

\begin{proof}
\textit{Necessity.}
The first inequality follows from
$1-\varepsilon_Q=\operatorname{Tr}(P\rho P)\ge d_Pa_0$.
For the coherence bound, positivity of
$\rho=\bigl(\begin{smallmatrix}A&B\\B^*&C\end{smallmatrix}\bigr)$ implies the
contraction factorization $B=A^{1/2}KC^{1/2}$ for some contraction $K$, hence
\[
\|B\|_F^2=\operatorname{Tr}\!\bigl(C^{1/2}K^*AKC^{1/2}\bigr)
\le\lambda_{\max}(A)\operatorname{Tr}C.
\]
Since $A\ge a_0P$, all eigenvalues of $A$ are at least $a_0$. Hence
\[
\operatorname{Tr}A=\sum_{k=1}^{d_P}\lambda_k(A)
\ge \lambda_{\max}(A)+(d_P-1)a_0,
\]
and, using $\operatorname{Tr}A=1-\varepsilon_Q$,
\[
\lambda_{\max}(A)\le1-\varepsilon_Q-(d_P-1)a_0=a_*.
\]
Therefore $c=\|B\|_F^2\le a_*\varepsilon_Q$.

\textit{Sufficiency.}
If the two displayed inequalities hold, choose an orthonormal basis $\{|p_k\rangle\}_{k=1}^{d_P}$
of $P\cH$ and a unit vector $|q_1\rangle\in Q\cH$. The state
\[
\rho_*:=
\begin{pmatrix}a_*&\sqrt{c}\\\sqrt{c}&\varepsilon_Q\end{pmatrix}_{|p_1\rangle,|q_1\rangle}
\oplus\;a_0\sum_{k=2}^{d_P}|p_k\rangle\langle p_k|
\oplus\;0_{Q\cH\ominus\mathbb{C}q_1}
\]
is positive semidefinite (the condition $c\le a_*\varepsilon_Q$ is exactly the Schur-complement
positivity of the active $2\times2$ block), has trace
$a_*+(d_P-1)a_0+\varepsilon_Q=1$, and satisfies all constraints of
$\mathcal{A}(a_0,\varepsilon_Q,c)$.
\end{proof}

\paragraph{Feasibility.}
Lemma~\ref{lem:admissible-feasibility} shows that the hypotheses $1-\varepsilon_Q\ge d_Pa_0$
and $0\le c\le a_*\varepsilon_Q$ are exactly the nonemptiness conditions for the admissible
class. They are also precisely the Schur-complement positivity conditions for the active
$2\times2$ block of the optimizer.

\begin{theorem}[Sharp entropy cost of coherent leakage / One-channel extremizer]
\label{thm:variational}

\noindent\textit{Standing notation for this theorem.}
$\mathcal{H}=P\mathcal{H}\oplus Q\mathcal{H}$ is a finite-dimensional Hilbert space.
$\Pi\rho=P\rho P\oplus Q\rho Q$ is the block-diagonal pinching.
$D(\rho\|\Pi\rho)=\operatorname{Tr}[\rho(\log\rho-\log\Pi\rho)]$ is the
block-coherence entropy.
The \emph{admissible class} is
\[
\mathcal{A}(a_0,\varepsilon_Q,c)
:=\bigl\{\rho\ge0,\,\operatorname{Tr}\rho=1:
P\rho P\ge a_0P,\;
\operatorname{Tr}(Q\rho Q)=\varepsilon_Q,\;
\|P\rho Q\|_F^2=c\bigr\}.
\]
This class holds the block weights (protected-sector floor $a_0$, leakage mass
$\varepsilon_Q$) and the Hilbert--Schmidt coherence magnitude $c$ fixed, and varies all
other degrees of freedom. The entropy $D(\rho\|\Pi\rho)$ is finite for every
$\rho\in\mathcal{A}$ with $A,C>0$, and extends to boundary states by continuity of
$\Phi$ (Lemma~\ref{lem:Phi-continuity}).

\smallskip\noindent\textit{Hypotheses.}
Fix $a_0>0$, $\varepsilon_Q\ge0$, $c\ge0$ with $d_P\ge1$, $d_Q\ge1$,
$1-\varepsilon_Q\ge d_Pa_0$, and $c\le a_*\varepsilon_Q$.

\smallskip\noindent\textit{Conclusion.}
\begin{equation}\label{eq:variational}
\inf_{\rho\in\mathcal{A}(a_0,\varepsilon_Q,c)}D(\rho\|\Pi\rho)\;=\;\Phi(a_*,\varepsilon_Q,c).
\end{equation}
The infimum can be attained by the explicit one-channel optimizer
\begin{equation}\label{eq:rho-star}
\rho_*:=
\begin{pmatrix}a_*&\sqrt{c}\\\sqrt{c}&\varepsilon_Q\end{pmatrix}_{|p_1\rangle,|q_1\rangle}
\oplus\;a_0\sum_{k=2}^{d_P}|p_k\rangle\langle p_k|\oplus 0_{Q\ominus\mathbb{C}q_1},
\end{equation}
where $\{|p_k\rangle\}_{k=1}^{d_P}$ is an orthonormal basis of $P\cH$,
$|q_1\rangle\in Q\cH$ is a unit vector, and $0_{Q\ominus\mathbb{C}q_1}$
denotes the zero operator on the subspace of $Q\cH$ orthogonal to $|q_1\rangle$.
All displayed direct sums are taken relative to these named active vectors and
their orthogonal complements within $P\cH$ and $Q\cH$ respectively.
\end{theorem}

\begin{proof}
\textit{Lower bound.}
For any feasible $\rho$, write the SVD of $B=P\rho Q$ as $B=\sum_{j=1}^m s_j|u_j\rangle\langle v_j|$,
where $\{u_j\}\subset P\cH$ are left singular vectors and $\{v_j\}\subset Q\cH$ are right singular
vectors. Let $a_j:=\langle u_j,Au_j\rangle$, $c_j:=\langle v_j,Cv_j\rangle$, $x_j:=s_j^2$.
Complete $\{u_j\}_{j=1}^m$ to an orthonormal basis $\{u_j\}_{j=1}^{d_P}$ of $P\cH$,
setting $c_j=x_j=0$ for $j>m$.
By the SVD pinching reduction (Appendix~\ref{app:svd}),
$D(\rho\|\Pi\rho)\ge\sum_j\Phi(a_j,c_j,x_j)$.

Write $a_j=a_0+\delta_j$ where $\delta_j:=a_j-a_0\ge0$ is the excess of the $j$-th
diagonal entry above the floor. Since $\{u_j\}_{j=1}^{d_P}$ is an orthonormal basis of $P\cH$,
\[
\sum_{j=1}^{d_P} a_j=\operatorname{Tr}_P(A)=1-\varepsilon_Q,
\]
so $\sum_j\delta_j=(1-\varepsilon_Q)-d_Pa_0$ and
$a_0+\sum_j\delta_j=1-\varepsilon_Q-(d_P-1)a_0=a_*$.
Let $\varepsilon_{\mathrm{rem}}:=\varepsilon_Q-\sum_jc_j$.
To confirm $\varepsilon_{\mathrm{rem}}\ge0$: since $\{v_j\}_{j=1}^m$ is an orthonormal set
in $Q\cH$ and $C=Q\rho Q$,
\[
\sum_{j=1}^mc_j=\sum_{j=1}^m\langle v_j,Cv_j\rangle\le\operatorname{Tr}(C)=\varepsilon_Q.
\]
Theorem~\ref{thm:compression} applies with excesses $\delta_j$, merged
protected mass $A=a_*$, merged leakage $E=\varepsilon_Q$, and total coherence $X=c\le a_*\varepsilon_Q=AE$.
Therefore $\sum_j\Phi(a_j,c_j,x_j)\ge\Phi(a_*,\varepsilon_Q,c)$.

\textit{Attainment.}
The condition $c\le a_*\varepsilon_Q$ ensures $\rho_*$ in~\eqref{eq:rho-star} is positive
semidefinite. One verifies $P\rho_*P\ge a_0P$, $\operatorname{Tr}(Q\rho_*Q)=\varepsilon_Q$,
$\|P\rho_*Q\|_F^2=c$. The spectator blocks are identical in $\rho_*$ and $\Pi\rho_*$,
contributing zero relative entropy. Evaluating the active block yields
$D(\rho_*\|\Pi\rho_*)=\Phi(a_*,\varepsilon_Q,c)$.
\end{proof}

\begin{remark}[Sharpness and limitations of the logarithmic specialization]
\label{rem:log-sharpness}
The exact variational minimum is $\Phi(a_*,\varepsilon_Q,c)$, but this does
not imply asymptotic equality in Corollary~\ref{cor:log-bound} merely by
letting $a_*\to a_0$. To compare the two bounds, fix $d_P\ge1$ and
$a_0>0$ with $d_Pa_0<1$, and take
$c=\kappa a_*\varepsilon_Q$ with fixed $0<\kappa\le1$ and
$a_*=1-\varepsilon_Q-(d_P-1)a_0$. The two eigenvalues defining $\Phi$
give, as $\varepsilon_Q\downarrow0$,
\[
\Phi(a_*,\varepsilon_Q,c)
=\frac{c}{a_*}\log\!\frac{a_*}{\varepsilon_Q}
+O(\varepsilon_Q),\qquad
\frac{\Phi(a_*,\varepsilon_Q,c)}{c\log(a_0/\varepsilon_Q)}
\longrightarrow\frac{1}{1-(d_P-1)a_0}.
\]
For $d_P=1$ the limit is $1$, recovering the asymptotic sharpness
exhibited by Proposition~\ref{prop:sharpness}. For fixed $a_0>0$ and
$d_P\ge2$ the limit is strictly greater than $1$: the uniform-floor
logarithmic bound remains valid but is not asymptotically saturated
along this optimizer family.
\end{remark}

\begin{remark}[Non-uniqueness and equality classification]
\label{rem:nonunique}
The optimizer is not unique in general. At $c=0$, every feasible block-diagonal
state is a minimizer. At $c>0$, varying the phase of the nonzero off-block entry
in the explicit two-level optimizer produces distinct feasible minimizers with
the same entropy cost. Unitaries acting only on its scalar or zero spectator
blocks leave that particular optimizer unchanged.
Whether these exhaust all minimizers is an open problem.

A complete classification of equality cases is not proved here. The explicit
two-level family is sufficient for attainment of the infimum. Necessity would
require equality analysis for the SVD pinching and the floor-aware merging CPTP
map, via the Petz equality theory~\cite{PetzSufficiency1988,Ruskai2002,HaydenJozsaPetzWinter2004}.
\end{remark}

The following corollary provides a sufficient condition for equality.

\begin{corollary}[Sufficient condition for equality]
\label{cor:equality}
Under the conditions of Theorem~\ref{thm:variational}, if $\rho$ has the form
\begin{equation}\label{eq:eq-char}
\rho=\begin{pmatrix}a_*&\sqrt{c}\,e^{i\phi}\\\sqrt{c}\,e^{-i\phi}&\varepsilon_Q\end{pmatrix}_{|p_1\rangle,|q_1\rangle}
\oplus\;a_0\sum_{k=2}^{d_P}|p_k\rangle\langle p_k|\oplus 0_{Q\ominus\mathbb{C}q_1}
\end{equation}
for some $\phi\in\mathbb{R}$, then $D(\rho\|\Pi\rho)=\Phi(a_*,\varepsilon_Q,c)$.
Whether every minimizer has this form is an open problem.
\end{corollary}

\begin{proof}
The phase $e^{i\phi}$ does not affect $D(\rho\|\Pi\rho)$: both the active $2\times2$
block and its pinching have phase-independent spectra. The spectator blocks in~\eqref{eq:eq-char}
are incoherent and identical in $\rho$ and $\Pi\rho$, so their relative entropy contribution is
zero. The result follows by evaluating $D$ on the active block alone.
\end{proof}

\subsection{Boundary entropy modulus}

\begin{corollary}[Boundary scaling modulus]
\label{cor:modulus}
Under the conditions of Theorem~\ref{thm:variational},
\[
D(\rho\|\Pi\rho)\ge\Phi(a_*,\varepsilon_Q,c).
\]
For an admissible family with $c(\varepsilon_Q)=\tau a_*\varepsilon_Q$,
$\tau\in(0,1]$ fixed, as $\varepsilon_Q\to0$ with $d_P,a_0$ fixed and
$a_*=1-\varepsilon_Q-(d_P-1)a_0\to1-(d_P-1)a_0>0$:
\[
\Phi(a_*,\varepsilon_Q,c(\varepsilon_Q))\sim\tau\varepsilon_Q\log(a_*/\varepsilon_Q)\to0.
\]
The total variational value tends to zero, but the normalized cost per unit coherence
exhibits logarithmic amplification:
\[
\frac{\Phi(a_*,\varepsilon_Q,c(\varepsilon_Q))}{c(\varepsilon_Q)}
=\frac{\Phi(a_*,\varepsilon_Q,c(\varepsilon_Q))}{\tau a_*\varepsilon_Q}
\sim\frac{\log(a_*/\varepsilon_Q)}{a_*}\to+\infty.
\]
Near the support boundary, each unit of coherence $c(\varepsilon_Q)$ costs a
diverging entropy factor.
\end{corollary}

\section{Entropy Dissipation Along the Block-Dephasing Orbit}
\label{sec:dynamics}

The static inequality (Theorem~\ref{thm:BKM-lower}) implies a lower
bound on entropy dissipation when the off-block coherence decays at a uniform rate.
The orbit below is an idealized uniform block-dephasing flow: it removes the
already-present off-block coherence exponentially while keeping the block-diagonal
populations fixed. It is generated by
$\mathcal{L}(\sigma)=-\Gamma(P\sigma Q+Q\sigma P)$ and is not intended as a
microscopic model of population transfer into $Q\mathcal{H}$.

\begin{corollary}[Dephasing orbit dissipation bound]
\label{cor:dephasing}
Let $M=\Pi\rho$, $Y=\rho-\Pi\rho$, and $\Gamma>0$. Consider the exact dephasing orbit
\[
\rho_t:=M+e^{-\Gamma t}Y,\quad t\ge0,
\]
generated by $\mathcal{L}(\sigma)=-\Gamma(P\sigma Q+Q\sigma P)$.
Assume $M\pm Y\ge0$ and $M>0$. Then for every $t>0$ (and also at
$t=0$, with the right derivative, if $\rho_0=M+Y>0$),
\begin{equation}\label{eq:dephasing-prod}
-\frac{d}{dt}D(\rho_t\|\Pi\rho)\;\ge\;2\Gamma e^{-2\Gamma t}\operatorname{Tr}[B^*\Omega_{A,C}^{-1}(B)].
\end{equation}
\end{corollary}

\begin{proof}
Set $\alpha=e^{-\Gamma t}$ and
\[
D(\alpha):=D(M+\alpha Y\|M).
\]
Since $M=\Pi\rho$ is fixed along the orbit, $D(\rho_t\|M)=D(\rho_t\|\Pi\rho)$.
For $0\le\alpha<1$, $M+\alpha Y=(1-\alpha)M+\alpha(M+Y)>0$,
so the following differentiations are valid. They also hold at $\alpha=1$
when $M+Y>0$. Thus
\[
D'(\alpha)=
\operatorname{Tr}\!\left[
Y\bigl(\log(M+\alpha Y)+I-\log M\bigr)
\right],
\]
so
\[
D'(0)=\operatorname{Tr}(Y)=0,
\]
with the same block-diagonal/off-diagonal cancellation as in Theorem~\ref{thm:BKM-lower}.
Taylor's formula gives
\[
D(\alpha)=\int_0^\alpha(\alpha-s)H_{M+sY}(Y,Y)\,ds.
\]
Thus $D'(\alpha)=\int_0^\alpha H_{M+sY}(Y,Y)\,ds$.
By Lemma~\ref{lem:midpoint-mono}, $H_{M+sY}(Y,Y)\ge H_M(Y,Y)$, giving
$D'(\alpha)\ge\alpha H_M(Y,Y)$.
Since $\dot\alpha=-\Gamma\alpha$,
\[
-\frac{d}{dt}D(\rho_t\|M)=\Gamma\alpha\,D'(\alpha)\ge\Gamma\alpha^2 H_M(Y,Y)
=2\Gamma e^{-2\Gamma t}\operatorname{Tr}[B^*\Omega_{A,C}^{-1}(B)].
\qedhere
\]
\end{proof}

\begin{corollary}[Boundary logarithmic enhancement]
\label{cor:log-dephasing}
Under the hypotheses of Corollary~\ref{cor:dephasing} (in particular $M>0$,
hence $A,C>0$), if additionally $A\ge a_0P$
and $0<\varepsilon_Q:=\operatorname{Tr}C\le a_0/2$, then for every
$t>0$ (and also at $t=0$, with the right derivative, if $\rho_0=M+Y>0$),
\[
-\frac{d}{dt}D(\rho_t\|\Pi\rho)
\;\ge\;2\Gamma e^{-2\Gamma t}\|B\|_F^2\log\!\left(\frac{a_0}{\varepsilon_Q}\right).
\]
\end{corollary}

\begin{proof}
The hypotheses of Corollary~\ref{cor:dephasing} include $M>0$, which implies $A,C>0$
and $\varepsilon_Q>0$.

\emph{Step 1.} By Corollary~\ref{cor:dephasing},
\[
-\frac{d}{dt}D(\rho_t\|\Pi\rho)
\;\ge\;
2\Gamma e^{-2\Gamma t}\operatorname{Tr}[B^*\Omega_{A,C}^{-1}(B)].
\]

\emph{Step 2.} Using the channel-weighted representation
(Proposition~\ref{prop:channel-weight}),
\[
\operatorname{Tr}[B^*\Omega_{A,C}^{-1}(B)]
=\sum_{\alpha,\beta}|\langle e_\alpha,Bf_\beta\rangle|^2\,L(a_\alpha,c_\beta).
\]
Under $A\ge a_0P$ and $0<\varepsilon_Q\le a_0/2$, every active channel
satisfies $a_\alpha\ge a_0$, $0<c_\beta\le\varepsilon_Q$, and $0<a_\alpha-c_\beta\le1$.
Hence
\[
L(a_\alpha,c_\beta)=\frac{\log(a_\alpha/c_\beta)}{a_\alpha-c_\beta}
\ge\log(a_\alpha/c_\beta)\ge\log(a_0/\varepsilon_Q).
\]
Summing over active channels gives
$\operatorname{Tr}[B^*\Omega_{A,C}^{-1}(B)]\ge\|B\|_F^2\log(a_0/\varepsilon_Q)$.

Combining Steps 1 and 2 gives the stated estimate.
\end{proof}

\begin{remark}[Relation to the modified logarithmic Sobolev inequality (MLSI)]
Corollary~\ref{cor:dephasing} is a one-orbit dissipation estimate for the exact
uniform dephasing flow. A full modified logarithmic Sobolev inequality would require
the reverse type of control, namely an upper bound on $D(\rho\|\Pi\rho)$ in terms of
the associated entropy production. That stronger problem is related to
recent work~\cite{KastoryanoTemme2013,CarlenMaas2017,BardetCapelLuciaPerezGarciaRouze2019,WirthZhang2021}.
\end{remark}

\section{Application: Coherent Leakage from a Protected Subspace}
\label{sec:application}

The preceding estimates give a kinematic certificate for coherent leakage from a
protected subspace. Suppose that $P\mathcal{H}$ is an intended protected,
computational, or metastable subspace and that $Q\mathcal{H}$ is its complementary
sector. For
\[
\rho=\begin{pmatrix}A&B\\ B^*&C\end{pmatrix},\qquad
A=P\rho P,\quad B=P\rho Q,\quad C=Q\rho Q,
\]
the scalar $\varepsilon_Q:=\operatorname{Tr}(C)$ is the population outside the
protected subspace, while $B=P\rho Q$ records coherent phase information between
the two sectors. The pinching $\Pi\rho=A\oplus C$ erases this cross-sector
coherence while preserving the block populations. Thus $D(\rho\|\Pi\rho)$ is
the information-theoretic cost of erasing protected--leakage coherence once the
state $\rho$ and the decomposition $P\oplus Q$ are specified. The certificate is
directly operational when $P\rho Q$ can be estimated, for example from tomography,
leakage-level interferometry, or numerical simulation. If only the population
$\varepsilon_Q$ is known, the theorem alone gives no nonzero coherent-leakage
certificate, since $B=0$ is compatible with the same leakage population.

For example, in a qutrit leakage model with
$P\mathcal{H}=\operatorname{span}\{|0\rangle,|1\rangle\}$ and
$Q\mathcal{H}=\operatorname{span}\{|2\rangle\}$, the scalar
$\varepsilon_Q=\langle2|\rho|2\rangle$ measures population outside the
computational subspace, while
\[
P\rho Q=\begin{pmatrix}\rho_{02}\\ \rho_{12}\end{pmatrix}
\]
records the coherent protected--leakage amplitudes. The certificate distinguishes
such coherent leakage from a block-diagonal state with the same $\varepsilon_Q$
but $P\rho Q=0$.

The scalar boundary estimate gives a simple certificate whenever the protected
block has the uniform floor $A\ge a_0P$:
\[
D(\rho\|\Pi\rho)
\ge
\|P\rho Q\|_F^2\log\!\left(\frac{a_0}{\varepsilon_Q}\right),
\qquad 0<\varepsilon_Q\le a_0/2.
\]
Since positivity forces $\|P\rho Q\|_F^2$ to vanish with $\varepsilon_Q$, this
estimate does not assert divergence of the total entropy cost. It shows instead
that, whenever $P\rho Q\ne0$,
\[
\frac{D(\rho\|\Pi\rho)}{\|P\rho Q\|_F^2}
\ge
\log\!\left(\frac{a_0}{\varepsilon_Q}\right),
\]
so the cost per unit Hilbert--Schmidt off-block coherence grows logarithmically
as the leaked population becomes small.

When the uniform floor is unavailable, the operator-channel estimate is the more
informative diagnostic:
\[
D(\rho\|\Pi\rho)\ge
\operatorname{Tr}\!\left[B^*\Omega^{-1}_{A,C}(B)\right].
\]
In the spectral representation of $A$ and $C$, this bound weights the actual
channels carrying $B$ by the BKM logarithmic kernel. It distinguishes, for example,
coherence supported on a well-populated protected direction from coherence
supported near the boundary of the protected block.

The variational theorem adds the corresponding extremal statement at fixed scalar
data. At fixed floor $A\ge a_0P$, fixed leakage population $\varepsilon_Q$, and
fixed quadratic coherence $c=\|P\rho Q\|_F^2$, the least entropy-costly coherent
leakage pattern can be realized by a single effective protected-to-leakage
transition. This is an existence statement for an extremizer, not a classification
of all minimizers. The dephasing orbit above has the parallel operational reading:
it describes active erasure of already-present protected--leakage coherence at
fixed block populations $A$ and $C$, rather than population transfer into the
leakage sector.

\section{The Midpoint BKM Mechanism for $\mathbb{Z}_2$-Gradings}
\label{sec:z2}

The exact one-channel variational theorem and logarithmic boundary law proved above
are specific to the finite-dimensional two-block pinching geometry.
The midpoint BKM mechanism, however, depends only on a sign symmetry: it uses the
fact that $\theta(M)=M$ and $\theta(Y)=-Y$ to force the BKM Hessian to be symmetric
around the block-diagonal point, together with the scalar resolvent midpoint inequality.
We record this extension separately: first for finite-dimensional trace-preserving
$\mathbb{Z}_2$-graded automorphisms, and then, under bounded strict-positivity
assumptions, for finite von Neumann algebras with faithful normal finite traces.

\subsection{Finite-dimensional $\mathbb{Z}_2$-graded midpoint BKM principle}

Throughout this subsection let $\mathcal{A}$ be a finite-dimensional $C^*$-algebra
with faithful trace $\operatorname{Tr}$, and let $\theta:\mathcal{A}\to\mathcal{A}$
be a trace-preserving involutive $*$-automorphism ($\theta^2=\mathrm{id}$,
$\operatorname{Tr}(\theta X)=\operatorname{Tr}(X)$, $\theta(X^*)=\theta(X)^*$,
$\theta(XY)=\theta(X)\theta(Y)$). Let $E:=\tfrac{1}{2}(\mathrm{id}+\theta)$ be
the fixed-point conditional expectation onto $\mathcal{A}^\theta=\{X:\theta(X)=X\}$.

\begin{theorem}[Finite-dimensional graded midpoint BKM coercivity]
\label{thm:graded-fd}
Let $\rho>0$ in $\mathcal{A}$, normalized or unnormalized, and set
\[
M:=E\rho,\qquad Y:=\rho-E\rho.
\]
Then $\theta(M)=M$, $\theta(Y)=-Y$, and
\[
M+Y=\rho>0,\qquad M-Y=\theta(\rho)>0.
\]
The estimate is
\begin{equation}\label{eq:graded-main}
D(M+Y\,\|\,M)\;\ge\;\tfrac{1}{2}H_M(Y,Y),
\end{equation}
i.e.\ $D(\rho\,\|\,E\rho)\ge\tfrac{1}{2}H_{E\rho}(\rho-E\rho,\rho-E\rho)$.

Moreover, for the uniform symmetry-dephasing orbit
$\rho_t:=M+e^{-\Gamma t}Y$, $\Gamma>0$, one has
\begin{equation}\label{eq:graded-dephasing}
-\frac{d}{dt}D(\rho_t\,\|\,M)
\;\ge\;
\Gamma e^{-2\Gamma t}H_M(Y,Y).
\end{equation}
The argument does not depend on normalization; in the normalized case $D$ is the usual
quantum relative entropy.
\end{theorem}

\begin{proof}
\emph{Positivity of the path.}
Since $M+Y=\rho$ and $M-Y=\theta(\rho)$,
\[
M+sY=\tfrac{1+s}{2}\,\rho+\tfrac{1-s}{2}\,\theta(\rho)>0
\qquad(0\le s\le1),
\]
and similarly $M-sY>0$. For $0<s<1$ these are convex combinations with
strictly positive coefficients, while at $s=0,1$ the endpoints are $M$, $\rho$,
or $\theta(\rho)$, all strictly positive.

\emph{Symmetry step.}
Since $\theta^2=\mathrm{id}$ and $E=\tfrac{1}{2}(\mathrm{id}+\theta)$:
\[
\theta(M)=\tfrac{1}{2}(\theta(\rho)+\rho)=M,
\quad
\theta(Y)=\theta(\rho)-M=-Y,
\quad
\theta(M+sY)=M-sY.
\]

\emph{BKM covariance under $\theta$.}
For any strictly positive $N$ and self-adjoint $Z$, functional calculus gives
$\theta(N^{-1})=(\theta(N))^{-1}$. Using multiplicativity $\theta(AB)=\theta(A)\theta(B)$
and trace-preservation $\operatorname{Tr}(\theta W)=\operatorname{Tr}(W)$:
\begin{align*}
H_{\theta(N)}(\theta Z,\theta Z)
&=\int_0^\infty\operatorname{Tr}\!\bigl[\theta(Z)\,(\theta(N)+rI)^{-1}\,
  \theta(Z)\,(\theta(N)+rI)^{-1}\bigr]\,dr\\
&=\int_0^\infty\operatorname{Tr}\!\bigl[\theta\!\bigl(Z(N+rI)^{-1}Z(N+rI)^{-1}\bigr)\bigr]\,dr
&&\text{(multiplicativity)}\\
&=H_N(Z,Z).
&&\text{(trace-preservation)}
\end{align*}
Applying this with $N=M+sY$ and $Z=Y$, and using $\theta(M+sY)=M-sY$, $\theta(Y)=-Y$:
\begin{equation}\label{eq:graded-symmetry}
H_{M-sY}(Y,Y)=H_{M+sY}(Y,Y).
\end{equation}

\emph{Resolvent midpoint estimate.}
Fix $r>0$. Set $R:=M+rI$ and $K:=R^{-1/2}YR^{-1/2}$ (self-adjoint). Then
$M\pm sY+rI=R^{1/2}(I\pm sK)R^{1/2}$, so
\[
\operatorname{Tr}\!\bigl[Y(M\pm sY+rI)^{-1}Y(M\pm sY+rI)^{-1}\bigr]
=\operatorname{Tr}\!\bigl[K^2(I\pm sK)^{-2}\bigr].
\]
Since $(I\pm sK)^{-2}$ is a Borel function of the self-adjoint operator $K$, it
commutes with $K^2$, and the scalar inequality
$\tfrac{1}{2}((1+u)^{-2}+(1-u)^{-2})\ge1$ (for $|u|<1$) applies by functional
calculus:
\[
\tfrac{1}{2}\bigl(\operatorname{Tr}[K^2(I+sK)^{-2}]+\operatorname{Tr}[K^2(I-sK)^{-2}]\bigr)
\ge\operatorname{Tr}[K^2].
\]
Integrating over $r>0$ and applying~\eqref{eq:graded-symmetry}:
\[
H_{M+sY}(Y,Y)=\tfrac{1}{2}\bigl(H_{M+sY}(Y,Y)+H_{M-sY}(Y,Y)\bigr)\ge H_M(Y,Y).
\]

\emph{Integration.}
Set $F(s):=D(M+sY\|M)$ for $s\in[0,1]$. Since $M+sY>0$ uniformly, $F$ is $C^2$
on $[0,1]$. We have $F(0)=0$. Moreover,
\[
F'(s)=\operatorname{Tr}\!\left[
Y\bigl(\log(M+sY)+I-\log M\bigr)
\right].
\]
Hence
\[
F'(0)=\operatorname{Tr}(Y)=0,
\]
because $E$ is trace-preserving and $Y=\rho-E\rho$.
Equivalently, $\operatorname{Tr}(YZ)=0$ for every $\theta$-invariant $Z$.
Taylor's formula with $F''(s)=H_{M+sY}(Y,Y)\ge H_M(Y,Y)$ gives
\[
F(1)=\int_0^1(1-s)F''(s)\,ds\;\ge\;\tfrac{1}{2}H_M(Y,Y),
\]
which is~\eqref{eq:graded-main}.

\emph{Dephasing.}
Set $z:=e^{-\Gamma t}$, so $D(\rho_t\|M)=F(z)$ and $\dot z=-\Gamma z$. Then
\[
-\frac{d}{dt}D(\rho_t\|M)=\Gamma z\,F'(z)
=\Gamma z\int_0^z H_{M+sY}(Y,Y)\,ds
\ge\Gamma z^2 H_M(Y,Y)
=\Gamma e^{-2\Gamma t}H_M(Y,Y),
\]
which is~\eqref{eq:graded-dephasing}.
\end{proof}

\paragraph{Convexity, involutive symmetry, and full matrix algebras.}
Theorem~\ref{thm:graded-fd} uses the same convexity and evenness mechanism
as Lemma~\ref{lem:midpoint-mono}; its generality is algebraic, not a
new convexity phenomenon. In $M_n(\mathbb{C})$ every $*$-automorphism is
inner. If $\theta^2=\mathrm{id}$ and $\theta=\operatorname{Ad}_U$,
then $U^2$ is central, hence a scalar phase. Rephasing $U$ gives
$U^2=I$, so $U=U^*$ and $U=P-Q$ for complementary orthogonal
projections $P,Q$. Thus the full-matrix case reduces precisely to the
Section~\ref{sec:BKM-geometry} two-block symmetry.
For a direct sum $\bigoplus_jM_{n_j}(\mathbb{C})$, the induced permutation
of minimal central projections is an involution. Each orbit is either a fixed
summand or a transposition of equal-size summands; trace preservation forces
the paired trace weights to agree. On a fixed summand, the restriction of
$\theta$ is inner and implemented by a self-adjoint sign unitary $P_j-Q_j$.
On each transposed pair, a choice of matrix-unit identification $V$ realizes
$\theta$, in an ambient block-matrix representation, as conjugation by the self-adjoint flip
$\left(\begin{smallmatrix}0&V\\V^*&0\end{smallmatrix}\right)$ on the two
copies. Thus the trace and the BKM quadratic form decompose into a
trace-weighted sum of independent two-block sign-symmetry contributions.
The direct-sum theorem assembles these instances; it does not assert the
one-channel variational conclusion in a general graded algebra.

\subsection{Recovery of the two-block theorem}
\label{ssec:z2-recovery}

Taking $\theta(X)=UXU^*$ with $U=P-Q$ gives $E=\Pi$ (the two-block pinching). Then
\[
M=\Pi\rho=A\oplus C,
\qquad
Y=\rho-\Pi\rho=\begin{pmatrix}0&B\\B^*&0\end{pmatrix}.
\]
For this $M$ and $Y$,
\[
H_M(Y,Y)=2\operatorname{Tr}\!\bigl[B^*\Omega_{A,C}^{-1}(B)\bigr],
\]
(computed in the proof of Theorem~\ref{thm:BKM-lower}). Hence the graded estimate
\[
D(\rho\|\Pi\rho)\ge\tfrac{1}{2}H_{\Pi\rho}(\rho-\Pi\rho,\rho-\Pi\rho)
\]
recovers exactly
\[
D(\rho\|\Pi\rho)\ge\operatorname{Tr}[B^*\Omega_{A,C}^{-1}(B)].
\]

Thus the two-block theorem is the canonical block-matrix instance of the
$\mathbb{Z}_2$-graded midpoint principle.

\begin{remark}[Admissible $\theta$: examples]
Theorem~\ref{thm:graded-fd} applies to parity or unitary reflection
$\theta(X)=UXU^*$ (any self-adjoint unitary $U=U^{-1}$), finite-dimensional
fermionic parity $\theta(X)=(-1)^F X(-1)^F$, and discrete spin-flip by a
tensor-product unitary involution. Antiunitary symmetries (such as time reversal)
fall outside the theorem's scope. In each full-matrix example above, the
involution is of the form $\operatorname{Ad}_{P-Q}$ after choosing its
self-adjoint unitary implementer; direct-sum summand swaps are covered
by the flip construction described above.
\end{remark}

\subsection{Finite von Neumann algebra extension}
\label{ssec:z2-semifinite}

The analytic midpoint principle has a direct formulation on a finite von
Neumann algebra; no ambient semifinite algebra is needed.

\begin{theorem}[Graded BKM coercivity in a finite von Neumann algebra]
\label{thm:graded-semifinite}
Let $(\mathcal{N},\tau)$ be a finite von Neumann algebra with a faithful normal
finite trace $\tau$ ($\tau(\mathbf{1})<\infty$). Let
$\theta:\mathcal{N}\to\mathcal{N}$ be a trace-preserving involutive
$*$-automorphism. Suppose $N,Z\in\mathcal{N}$ are bounded self-adjoint
operators with $\theta(N)=N$, $\theta(Z)=-Z$, and for constants
$0<m<M_0<\infty$,
\begin{equation}\label{eq:semifinite-bounds}
m\mathbf{1}\le N\pm Z\le M_0\mathbf{1}.
\end{equation}
Define
\[
D_\tau(A\|B)=\tau[A(\log A-\log B)],\qquad
H_N(Z,Z)=\int_0^\infty\tau[Z(N+r\mathbf{1})^{-1}Z(N+r\mathbf{1})^{-1}]\,dr.
\]
Then $\tau(Z)=0$, $H_N(Z,Z)<\infty$, and
\begin{equation}\label{eq:semifinite-main}
D_\tau(N+Z\|N)\ge\tfrac12 H_N(Z,Z).
\end{equation}
For $t\ge0$, let $\rho_t=N+e^{-\Gamma t}Z$ with $\Gamma>0$. Then
\begin{equation}\label{eq:semifinite-dephasing}
-\frac{d}{dt}D_\tau(\rho_t\|N)
\ge\Gamma e^{-2\Gamma t}H_N(Z,Z).
\end{equation}
The expression for $D_\tau$ is the unnormalized trace entropy; its
arguments have equal trace, so the standard unnormalized correction cancels.
\end{theorem}

\begin{proof}
Because $\tau\circ\theta=\tau$ and $\theta(Z)=-Z$,
$\tau(Z)=0$. The positivity bounds imply $m\mathbf{1}\le N+sZ\le M_0\mathbf{1}$
for $|s|\le1$. Since $\tau(\mathbf{1})<\infty$ and $Z$ is bounded,
$\tau(Z^2)<\infty$, and the resolvent integrand is bounded by
$\tau(Z^2)/(m+r)^2$; hence $H_N(Z,Z)<\infty$.

Set $F(s)=D_\tau(N+sZ\|N)$. Bounded functional calculus and the
resolvent integral for the Fr\'echet derivative of $\log$ imply
$F\in C^2([-1,1])$ with
\[
F'(s)=\tau[Z(\log(N+sZ)+\mathbf{1}-\log N)],\qquad
F''(s)=H_{N+sZ}(Z,Z).
\]
The trace symmetry gives $F'(0)=0$, while covariance under $\theta$
yields $H_{N+sZ}(Z,Z)=H_{N-sZ}(Z,Z)$. For each $r>0$ set
$R=N+r\mathbf{1}$ and $K=R^{-1/2}ZR^{-1/2}$. As in Lemma~\ref{lem:midpoint-mono},
\[
\frac12\tau[K^2(1+sK)^{-2}+K^2(1-sK)^{-2}]
\ge\tau[K^2],
\]
by functional calculus and positivity of the faithful finite trace.
Integrating in $r$ and using the involutive symmetry gives
$F''(s)\ge H_N(Z,Z)$ for $0\le s\le1$. Consequently,
\[
F(1)=\int_0^1(1-s)F''(s)\,ds\ge\tfrac12 H_N(Z,Z).
\]
Putting $z=e^{-\Gamma t}$, we have
$-d_tF(z)=\Gamma zF'(z)=\Gamma z\int_0^zF''(s)\,ds
\ge\Gamma z^2H_N(Z,Z)$, as claimed.
\end{proof}

\begin{remark}[Invariant finite-trace corners]
If $(\mathcal{M},\widetilde\tau)$ is semifinite, and $\theta$ is a
$\widetilde\tau$-preserving involutive $*$-automorphism with $\theta(e)=e$
for a projection $e$ satisfying $\widetilde\tau(e)<\infty$, then
$\theta$ restricts to a $\widetilde\tau$-preserving automorphism of the
finite von Neumann algebra $e\mathcal{M}e$. Applying
the theorem to this corner is an instance of the finite-algebra result,
not a separate extension to general semifinite densities.
\end{remark}

\subsection{Further directions}
\label{ssec:z2-open}

Theorem~\ref{thm:graded-semifinite} is a bounded finite-von-Neumann-algebra result.
It does not treat unbounded densities, general semifinite densities, type~III von
Neumann algebras, Araki relative entropy, non-uniform dephasing, full modified
logarithmic Sobolev inequalities, or a one-channel variational theory for general
$\mathbb{Z}_2$-gradings. These questions are left for future work.

\section{Conclusion}
\label{sec:discussion}

We have proved a set of finite-dimensional results for the entropy cost of coherent
leakage across a two-block decomposition. The central results include: (i) the BKM
midpoint coercivity estimate $D(\rho\|\Pi\rho)\ge\operatorname{Tr}[B^*\Omega_{A,C}^{-1}(B)]$;
(ii) the exact one-channel variational theorem identifying a single active two-level
block that attains the constrained minimum; and (iii) the previously known logarithmic boundary
law, here recovered in the operator-kernel setting~\cite{Nasreddine2026EntropyModuli}. The midpoint mechanism also admits controlled extensions to
finite-dimensional trace-preserving $\mathbb{Z}_2$-graded automorphisms and, under
bounded strict-positivity assumptions, to finite von Neumann algebras with a faithful normal finite trace (including invariant finite-trace corners of semifinite algebras).

\paragraph{Physical significance.}
The one-channel principle, in the sense proved here, says that the constrained minimum at prescribed
scalar coherence data can be realized by a qubit-like block, irrespective of the
ambient dimension. This is an information-theoretic constraint on coherent leakage in
settings where block-diagonal populations are fixed or controlled and cross-block
coherences carry quantum information cost.

\paragraph{Limitations.}
The one-channel variational theorem and logarithmic boundary barrier are proved only
for the two-block pinching geometry and are not extended here to multi-block
decompositions or general $\mathbb{Z}_2$-gradings. The present proof uses the
single block-sign involution and a two-level compression target; in a genuine
multi-block decomposition there is no single odd direction or single phase-alignment
constraint that merges all coherent channels at once. A complete classification of all
equality cases for the variational minimizer is not proved. Non-uniform dephasing and
full modified logarithmic Sobolev inequalities are not treated. The finite-von-Neumann-algebra extension is bounded and uniformly positive: extensions to unbounded densities, type~III
algebras, and Araki relative entropy require additional modular-theoretic arguments
and are left open.

This leaves several directions for further work: equality classification for the
compression mechanism via Petz equality theory~\cite{PetzSufficiency1988,Ruskai2002,HaydenJozsaPetzWinter2004}, multi-block and general
$\mathbb{Z}_2$-graded variational theory, non-uniform dephasing and MLSI, and
unbounded or type~III extensions.

\section*{AUTHOR DECLARATIONS}

\subsection*{Conflict of Interest}
The author has no conflicts to disclose.

\subsection*{Author Contributions}
Hassan Nasreddine: Conceptualization, Formal analysis, Writing -- original draft, Writing -- review and editing.

\section*{DATA AVAILABILITY}
Data sharing is not applicable to this article as no new data were created or analyzed in this study.
\appendix

\section{SVD Pinching and the BKM Operator}
\label{app:svd}

We give the explicit SVD pinching CPTP map used in Corollary~\ref{cor:log-bound}
and the proof of Theorem~\ref{thm:variational}.

\subsection*{Exact spectral formula}

Proposition~\ref{prop:channel-weight} establishes the exact eigenbasis formula:
with $A=\sum_\alpha a_\alpha|e_\alpha\rangle\langle e_\alpha|$ and
$C=\sum_\beta c_\beta|f_\beta\rangle\langle f_\beta|$,
\begin{equation}\label{eq:bkm-exact}
\operatorname{Tr}[B^*\Omega_{A,C}^{-1}(B)]=\sum_{\alpha,\beta}|\langle e_\alpha,Bf_\beta\rangle|^2L(a_\alpha,c_\beta).
\end{equation}
This uses the eigenbases of $A$ and $C$, not the singular vectors of $B$.

\subsection*{SVD pinching CPTP map}

This SVD-adapted pinching construction also appears in
Ref.~\cite{Nasreddine2026EntropyModuli}, Lemma 5.7 (Step 2).
We restate it to specify the CPTP map used in the present variational proof.

Let $B=U\Sigma V^*$ be the compact SVD, with $\Sigma=\diag(s_1,\ldots,s_k)$,
$U=[u_1|\cdots|u_k]$ on $P\mathcal{H}$, and $V=[v_1|\cdots|v_k]$ on $Q\mathcal{H}$.
Define rank-1 projections $P_j=|u_j\rangle\langle u_j|$ (on $P\mathcal{H}$) and
$Q_j=|v_j\rangle\langle v_j|$ (on $Q\mathcal{H}$), and let $R_P=P_{\ker B^*}$,
$R_Q=P_{\ker B}$. The \emph{SVD pinching channel} is
\begin{equation}\label{eq:svd-pinch}
\mathcal{E}(X):=\sum_{j=1}^k(P_j\oplus Q_j)X(P_j\oplus Q_j)+R_PXR_P+R_QXR_Q.
\end{equation}
The Kraus operators $K_j:=P_j\oplus Q_j$ are orthogonal projections, so
$K_j^*K_j=K_j$. The completeness condition
\[
\sum_j K_j^*K_j+R_P+R_Q=\sum_j(P_j\oplus Q_j)+R_P+R_Q=I
\]
holds by the orthogonal decomposition of $P\mathcal{H}$ and $Q\mathcal{H}$,
so $\mathcal{E}$ is trace-preserving and completely positive (CPTP).

\paragraph{Action on the state.}
For $\rho_{\mathrm{pin}}=\mathcal{E}(\rho)$:
$B_{\mathrm{pin}}=\Sigma$, $A_{\mathrm{pin}}=\diag(a_j^{\mathrm{pin}})$ with
$a_j^{\mathrm{pin}}=\langle u_j,Au_j\rangle$, $C_{\mathrm{pin}}=\diag(c_j^{\mathrm{pin}})$
with $c_j^{\mathrm{pin}}=\langle v_j,Cv_j\rangle$.
Since $\mathcal{E}$ commutes with $\Pi$: $\Pi(\rho_{\mathrm{pin}})=\mathcal{E}(\Pi\rho)$.

\paragraph{Data-processing inequality.}
By monotonicity of relative entropy under CPTP maps:
\[
D(\rho\|\Pi\rho)\ge D(\mathcal{E}(\rho)\|\mathcal{E}(\Pi\rho))=D(\rho_{\mathrm{pin}}\|\Pi\rho_{\mathrm{pin}}).
\]

\paragraph{Exact entropy decomposition.}
Since $\rho_{\mathrm{pin}}$ has diagonal blocks $A_{\mathrm{pin}}$, $C_{\mathrm{pin}}$,
and $B_{\mathrm{pin}}=\Sigma$ (all diagonal in the same orthonormal basis $\{u_j\}\cup\{v_j\}$),
the pinched state decomposes into orthogonal $2\times2$ blocks. The kernel blocks
are block-diagonal and contribute zero relative entropy. Therefore:
\[
D(\rho_{\mathrm{pin}}\|\Pi\rho_{\mathrm{pin}})
=\sum_{j=1}^k\Phi(a_j^{\mathrm{pin}},c_j^{\mathrm{pin}},s_j^2).
\]
Hence $D(\rho\|\Pi\rho)\ge\sum_{j=1}^k\Phi(a_j^{\mathrm{pin}},c_j^{\mathrm{pin}},s_j^2)$.
The BKM bound follows from $\Phi(a,c,x)\ge xL(a,c)$ (Lemma~\ref{lem:phi-bkm-lower}).

\section{Proof of the Polygon Modulus Lemma}
\label{app:polygon}

\begin{lemma}[Polygon modulus lemma]
\label{lem:polygon}
Let $\ell_1,\ldots,\ell_k\ge0$ satisfy $\max_j\ell_j\le T\le\sum_j\ell_j$.
Then there exist phases $\theta_j\in\mathbb{R}$ such that $|\sum_je^{i\theta_j}\ell_j|=T$.
\end{lemma}

\begin{proof}
Set $S=\sum_j\ell_j$ and $\ell_1=\max_j\ell_j$.
The map $\boldsymbol\theta\mapsto|\sum_je^{i\theta_j}\ell_j|$ is continuous on the compact
connected set $[0,2\pi]^k$, so its image is a closed interval $[m,S]$.

\textit{Maximum:} achieved at $\theta_j=0$ for all $j$, giving $S$.

\textit{Minimum:} We claim $m=\max(0,2\ell_1-S)$.
The lower bound $m\ge2\ell_1-S$ follows from the triangle inequality.
If $2\ell_1\ge S$, set $\theta_1=0$ and $\theta_j=\pi$ for $j\ge2$, giving
$|\ell_1-\sum_{j\ge2}\ell_j|=2\ell_1-S$.
If $2\ell_1<S$, the classical polygon-closure theorem
(\emph{non-negative reals $\ell_1,\ldots,\ell_k$ are side lengths of a closed polygon
iff $\max_j\ell_j\le\sum_{i\ne j}\ell_i$}) applies since
$\ell_1<S-\ell_1=\sum_{j\ge2}\ell_j$; this gives phases with $|\sum_j\ell_je^{i\theta_j}|=0$,
so $m=0$.

\textit{Conclusion:} $\max(0,2\ell_1-S)\le T\le S$, so $T\in[m,S]$.
\end{proof}

\section{Boundary Cases for Coherence Compression}
\label{app:edge}

The construction in Theorem~\ref{thm:compression} handles degenerate cases as follows.

\paragraph{Case $X=0$.} All $x_j=0$, so $B_j=0$. Both $M$ and $D$ are block-diagonal
and equal, giving $D(M\|D)=0=\Phi(A,E,0)$. The inequality holds with equality.

\paragraph{Case $\varepsilon_j=0$ for some $j$.}
Positivity of the corresponding $2\times2$ block implies $x_j\le a_j\varepsilon_j=0$,
hence $x_j=0$ and $\Phi(a_j,0,0)=0$. These channels contribute no relative entropy
and may be dropped from the sum.

\paragraph{Case $E=0$.}
$E=\varepsilon_{\mathrm{rem}}+\sum_j\varepsilon_j=0$ forces $\varepsilon_j=0$ for all $j$
and $\varepsilon_{\mathrm{rem}}=0$. Positivity then forces $x_j\le a_j\varepsilon_j=0$,
so $x_j=0$ for all $j$ and $X=0$. This is subsumed by the $X=0$ case.

\paragraph{Boundary of support.} When blocks approach the support boundary, the
inequality $\sum_j\Phi(a_j,\varepsilon_j,x_j)\ge\Phi(A,E,X)$ extends by continuity:
$\Phi$ is continuous on $\{a,\varepsilon\ge0,\,0\le x\le a\varepsilon\}$ (Lemma~\ref{lem:Phi-continuity})
and the data-processing inequality is preserved under limits of strictly positive
inputs (by standard approximation).

\section{Petz Midpoint Monotonicity: Complete Proof}
\label{app:petz}

This appendix is independent of the main entropy estimates. It records a parallel
midpoint monotonicity statement for Petz monotone metrics, to clarify the geometric
origin of the BKM argument and show that the coercivity mechanism is not specific
to the BKM metric. None of the main results in Sections~\ref{sec:setup}--\ref{sec:z2} depends on this appendix.

\begin{proposition}[Midpoint monotonicity for Petz metrics]
Let $g^f$ be a Petz monotone metric with symmetric normalized operator-monotone $f$.
Let $M>0$ be block-diagonal, $Y$ off-diagonal, with $M\pm tY>0$ for $t\in[0,1)$. Then
$g^f_{M+tY}(Y,Y)\ge g^f_M(Y,Y)$ for all $t\in[0,1)$.
\end{proposition}

\begin{proof}
\emph{Setup.}
Here $\langle\cdot,\cdot\rangle_{\mathrm{HS}}$ denotes the Hilbert--Schmidt inner product.
The Petz metric is $g^f_N(X,X)=\langle X,(K^f_N)^{-1}X\rangle_{\mathrm{HS}}$
where $K^f_N=m_f(L_N,R_N)$ is the Kubo--Ando mean of left and right multiplication superoperators.

\emph{Midpoint comparison.}
Set $N_\pm=M\pm tY$. Since $N_-=UN_+U^*$ with $UYU^*=-Y$, unitary covariance gives
$K^f_{N_-}=\mathcal{U}(K^f_{N_+})$ where $\mathcal{U}$ is conjugation by $U$.
Hence $g^f_{N_-}(Y,Y)=g^f_{N_+}(Y,Y)$.

\emph{Concavity step.}
$M=\tfrac{1}{2}(N_++N_-)$. By joint concavity of operator means~\cite{KuboAndo1980,Bhatia2007}:
\[
K^f_M\ge\tfrac{1}{2}(K^f_{N_+}+K^f_{N_-}).
\]

\emph{Conclusion.}
Since $t\mapsto t^{-1}$ is operator antimonotone then operator convex on $(0,\infty)$:
\[
(K^f_M)^{-1}\le\bigl(\tfrac{1}{2}K^f_{N_+}+\tfrac{1}{2}K^f_{N_-}\bigr)^{-1}
\le\tfrac{1}{2}(K^f_{N_+})^{-1}+\tfrac{1}{2}(K^f_{N_-})^{-1}.
\]
Hence $g^f_M(Y,Y)\le\tfrac{1}{2}[g^f_{N_+}(Y,Y)+g^f_{N_-}(Y,Y)]=g^f_{N_+}(Y,Y)$.
\end{proof}

\bibliographystyle{aipnum4-2}
\bibliography{references}

\end{document}